\documentclass[10pt,german]{article}
\usepackage[margin=1.4in]{geometry}

\usepackage{multicol}
\usepackage[font=small,labelfont=bf]{caption}

\usepackage{parskip}

\usepackage{graphicx}
\usepackage{natbib}

\usepackage{amssymb}
\usepackage{amsthm}
\usepackage[nosumlimits]{amsmath}
\usepackage{amsbsy}
\usepackage{bm}

\usepackage{graphicx}

\usepackage{mystuff}

\newtheorem{theorem}{Theorem}

\newtheorem{corollary}{Corollary}

\newtheorem{lemma}{Lemma}

\usepackage{setspace}

\DeclareGraphicsExtensions{.eps}

\title{Hamiltonian dynamics of several rigid bodies\\
interacting point vortices}
\author{Steffen Wei{\ss}mann\thanks{Institut f\"{u}r Mathematik, TU Berlin.
Email: steffen.weissmann@mail.de}}

\begin{document}

\maketitle

\begin{abstract}
We derive the dynamics of several rigid bodies of arbitrary shape in
a 2--dimensional inviscid and incompressible fluid, whose vorticity field is
given by point vortices. We adopt the idea of \citet{Vankerschaver2009} to
derive the Hamiltonian formulation via symplectic reduction from a canonical
Hamiltonian system. The reduced system is described by a non-canonical
symplectic form, which has previously been derived for a single, circular disk
using heavy differential-geometric machinery in an infinite-dimensional
setting.
In contrast, our derivation makes use of the fact that the dynamics of the
fluid, and thus the point vortex dynamics, is determined from first principles.
Using this knowledge we can directly determine the dynamics on the reduced,
finite-dimensional phase space, using only classical mechanics.
Furthermore, our approach easily handles several bodies of arbitrary shapes.
From the Hamiltonian description we derive a Lagrangian formulation, which
enables the system for variational time integrators. We briefly describe how
to implement such a numerical scheme and simulate different configurations for
validation.
\end{abstract}

\section{Introduction}
The Hamiltonian dynamics of a single rigid body of arbitrary shape in a
2--dimensional inviscid and incompressible fluid interacting with $n$ point
vortices has first been formulated by \citet{Shashikanth2005}. The system was
also studied by \citet{Borisov2007}, dropping the restriction of zero
circulation around the cylinder. Conceptually, both works use a momentum balance
approach to derive the equations of motion, i.e., changes in fluid momentum are
compensated by the body. This approach is inherently restricted to a single
rigid body, since it is not clear how to distribute changes in fluid momentum
over several bodies.\\
A different approach was taken by \citet{Vankerschaver2009}, who derived the
dynamics for the case of a single circular disk by considering the dynamics as
geodesics on a Riemannian manifold, in the spirit of Arnold's geometric
description of fluid dynamics \citep{Arnold1966}. The manifold here is the
Cartesian product of $\SE(2)$ with a subset of volume-preserving embeddings of
the initial fluid configuration into $\R^2$, compatible with the time-dependent
pose of the body. The Riemannian metric is given by the kinetic energy.
When reducing the system to fluid velocity fields generated by point vortices,
one obtains a finite-dimensional phase space with magnetic symplectic form,
which yields the coupling between rigid body and point vortex motion.
While in principle it is possible to extend this to several bodies of arbitrary
shape, the derivation is challenging and requires heavy differential-geometric
machinery in an infinite-dimensional setting: One has to determine the curvature
of the mechanical connection on unreduced phase space, which is already
challenging for a single, circular disk.\\
Our derivation makes use of the fact that the dynamics of the fluid, and thus
the point vortex dynamics, is already known. Using this knowledge we can
directly determine the dynamics on the reduced, finite-dimensional phase space,
using only classical mechanics. The derivation readily handles the case of
several bodies of arbitrary shape.

The system that we study here can be viewed as the superposition of two simpler
and well-understood systems: Point vortex dynamics in the plane, and rigid body
dynamics in potential flow. In fact, as we will show later, at large distance
the two systems evolve independently.\\
The study of point vortex dynamics dates back to the seminal work by
\citet{Helmholtz1858}. Since then it has been an active area of research, see,
for instance, \citet{Saffman1992,Newton2001}. Apart from being a rich source for
mathematical research \citep{Aref2007}, point vortices are of great interest for
numerical simulation of fluid flow since \citet{Chorin1973}, supported by
strong analytical results \citep{Majda2002}.
The dynamics is governed by a Hamiltonian system which is non-canonical in the sense
that point vortex positions are already points in phase space. Physically, this
means that one cannot assign an initial velocity or momentum to the vortices,
their motion is determined completely from fluid dynamics.\\
The dynamics of several rigid bodies in potential flow (i.e., no vorticity) has
been studied by \citet{Nair2007}. Their work is based on \citet{Lamb1895},
also \cite{Milne1996} provides an extensive treatment of fluid-body interaction.
\cite{Kirchhoff1870} was the first to discover that the kinetic energy of a
surrounding potential flow can be incorporated into the kinetic energy of rigid
motion as \emph{added mass}. In contrast to the case of a single rigid body, the
kinetic energy of potential flow around several rigidly moving obstacles is no
longer a constant quadratic form on body velocities, but depends on the relative
poses of the different bodies. Still, the dynamics of this system is Hamiltonian
in a canonical way:
The kinetic energy defines a Riemannian metric on the configuration space, and
geodesics solve Hamilton's equations with respect to the canonical symplectic
form on the cotangent bundle, and kinetic energy as the Hamiltonian.

In this paper we introduce the Hamiltonian dynamics of several rigid bodies
interacting with point vortices, for the case of zero circulation around the
individual bodies, but arbitrary strengths of the point vortices.
The dynamics of this system has been known only for the case of a single rigid
body. In order to derive the equations of motion we adopt the description of the
reduced phase space from \citet{Vankerschaver2009} and extend it to the case of
several rigid bodies.
On the reduced phase space we determine the magnetic symplectic form directly,
using only general properties of the magnetic symplectic form, and first
principles of fluid dynamics.
From the Hamiltonian formulation we give a Lagrangian description of the system,
which enables the system for variational integrators \citep{Marsden2001}.\\
From the smooth Lagrangian description we briefly describe how to construct a
numerical scheme to simulate the time evolution of the system. The Lagrangian
here is degenerate, so the system fits into the framework of variational
integrators for degenerate Lagrangian systems \citep{Rowley2002}. We develop a
variational time integrator which captures the qualitative behavior of the
dynamics over long simulation times, has excellent energy behavior, and
preserves momentum and symplectic structure exactly. For validation we apply our
method to some integrable and chaotic configurations.

\section{Physical Model}

\subsection{Rigid Bodies}
\label{sec:rigid-bodies}

The motion of a rigid body is described by a time-dependent Euclidean
transformation
\begin{align*}
g \colon\, z & \mapsto R z + y, & R & =
\mat{\cos \theta & - \sin \theta \\ \sin \theta & \cos \theta},
\end{align*}
where $R \in \SO(2)$ is a $2 \times 2$-rotation matrix, $\theta \in \left[ 0, 2
\pi \right)$ specifies the angle of rotation, and $y = (y_1, y_2) \in
\R^2$ describes the location of the center of the body. It is convenient to
identify Euclidean transformations with $3 \times 3$-matrices, acting on homogeneous
vectors:
\begin{align}
\label{eqn:matrix-identification-g}
g \colon \quad z & \mapsto R z + y \quad \longleftrightarrow
\quad
\mat{R & y \\ 0 & 1} \colon \, \vec{z\\1} \mapsto \vec{R z + y \\ 1}.
\end{align}
Concatenation of Euclidean transformations becomes matrix multiplication in this
representation. The time derivative of $g$ can be expressed as
\begin{align}
\label{eqn:time-derivative-g}
\dot{g} & = \mat{\dot{R} & \dot{y} \\ 0 & 0} = \underbrace{\mat{R & y \\ 0 &
1}}_g \underbrace{\mat{\Omega \times & V \\ 0 & 0}}_\Xi, & \Omega \times & =
\mat{0 & -\Omega
\\
\Omega & 0},
\end{align}
where we have denoted the angular velocity by $\Omega = \dot{\theta}$ and
the linear velocity by $V = R^t \dot{y}$. The matrix $\Omega \times$ acts
as a stretched $90^\circ$ rotation, i.e., as cross-product with $\Omega e_3$.
Here $g$ transforms the velocity field
\begin{align}
\label{eqn:body-velocity}
\Xi \colon \quad z \mapsto \Omega \times z + V \quad \longleftrightarrow \quad
\mat{\Omega \times & V \\ 0 & 0} \colon \, \vec{z\\1} \mapsto \vec{\Omega \times
z + V \\ 1}
\end{align}
which we can also identify with $3 \times 3$-matrices. We call $\Xi$
the \emph{body velocity}, it represents the instantaneous velocity field
expressed in the body frame. By a change of variables (through conjugation with
$g$) we obtain from $\Xi$ the \emph{spatial velocity} $\xi$:
\begin{align*}
\xi = g \, \Xi \, g^{-1}
 = \mat{\Omega \times & R V + y \times \Omega \\ 0 & 0}
=: \mat{\omega \times & v \\ 0 & 0}.
\end{align*}
We will also identify body/spatial velocity with 3--vectors of angular and
linear velocity components: $\Xi = (\Omega, V)$, $\xi = (\omega, v)$.
In this representation, velocity conversion between body and spatial frame
corresponds to matrix multiplication:
\begin{align}
\label{eqn:adjoint-operator}
\xi & = g \, \Xi \, g^{-1} = \Ad{g} \, \Xi, & \Ad{g} & = \mat{I & 0 \\ y
\times & R}.
\end{align}
The kinetic energy of a moving rigid body is
\begin{align*}
T & = \frac{m}{2} \norm{V}^2 + \frac{\imath}{2} \norm{\Omega}^2,
\end{align*}
where $m$ is the body mass and $\imath$ the body's moment of inertia, i.e., its
resistance against changes in angular velocity. We can write $T$ as
\begin{align}
\label{eqn:rigid-body-kinetic-energy}
T & = \frac{1}{2} \dd{\MI \, \Xi, \Xi} = \frac12 \dd{M, \Xi}, & 
\text{ where } & & \MI & = \mat{\imath & 0 \\ 0 & m I}
\end{align}
is the \emph{mass-inertia tensor} of the body, mapping body velocity
$\Xi=(\Omega, V)$ to body momentum $M = (A, L)$. $A$ and $L$ denote
angular and linear momentum, respectively. As for velocity we can also express
momentum in the spatial frame, i.e., the spatial momentum $m = (a, \ell)$
satisfies $\tdd{m, \xi} = \tdd{M, \Xi}$. This gives
\begin{align}
\label{eqn:spatial-momentum}
m = \CoAd{g^{-1}} M = \mat{I & (y \times) R \\ 0 & R} \vec{A \\ L} = \vec{A
+ y \times R L \\ R L} =: \vec{a \\ \ell},
\end{align}
where $\CoAd{g^{-1}}$ denotes the matrix transpose of $\Ad{g^{-1}}$,
defined in Equation~\eqref{eqn:adjoint-operator}.

The space of rigid motions forms the Lie group $\SE(2)$. Using the
representation in terms of $3 \times 3$-matrices, the group law (i.e.,
concatenation of Euclidean transformations), is just matrix multiplication.
The tangent space $T_g \SE(2)$ at $g$ consists of elements of the form $\delta g
= g \Gamma$, where $\Gamma$ is a $3 \times 3$ matrix representing the body
velocity field $\Gamma: z \mapsto \Lambda \times z + U$, see
Equation~\eqref{eqn:body-velocity}. The space of such matrices (or velocity
vector fields) is the Lie algebra $\se(2)$, which we identify with $\R^3$
through $\Gamma = (\Lambda, U)$.\\
In order to express the equations of motion of rigid bodies and point vortices
we need the notion of the \emph{left gradient} of a scalar function $f(g)$.
The differential $D_g f$ is a linear map of tangent vectors $\delta g = g
\Gamma$ to the real numbers. It follows that $D_g f$ is also linear in $\Gamma$,
and we define the left gradient $\lgrad f(g)$ as the 3-vector which satisfies
\begin{align}
\label{eqn:left-gradient}
\dd{\lgrad f(g), \Gamma} = D_g f(\delta g).
\end{align}

\subsection{Fluid Configuration}

The time-dependent fluid domain $\Fld$ is covered with a fluid which is at rest
at infinity and whose motion is given by a time-dependent fluid velocity field
$u$. Its vorticity field $\omega = \curl u$ is zero everywhere, except for
isolated \emph{point vortices} $\gamma = \{\gamma_1,\ldots, \gamma_m\}$. There
the vorticity field is concentrated in a delta-function-like manner. The
circulation around each vortex is constant in time (due to Kelvin's circulation
theorem) and measures the \emph{strength} $K_i$ of the vortex $\gamma_i$.
We assume zero circulation around the individual bodies, and impose no-through
boundary conditions. That is, the normal component of the velocity field must
coincide with the body boundary normal velocity, while the tangent velocity is
arbitrary:
\begin{align}
\label{eqn:motion-boundary-condition}
\dd{u(z), n_j(z)} & = \dd{\xi_j(z), n_j(z)} = \dd{\omega_j \times z + v_j,
n_j(z)} \text{ for } z \in \partial {\Bdy}_j.
\end{align}
Here $\xi_j = (\omega_j, v_j)$ denotes the velocity field of $\Bdy_j$'s motion
in the spatial frame of reference, and $n_j$ is the normal vector field along $\partial
\Bdy_j$, also in the spatial frame.

\subsubsection{Hodge-Helmholtz Decomposition}

\begin{figure}
\centering
\includegraphics[width=0.4\columnwidth]{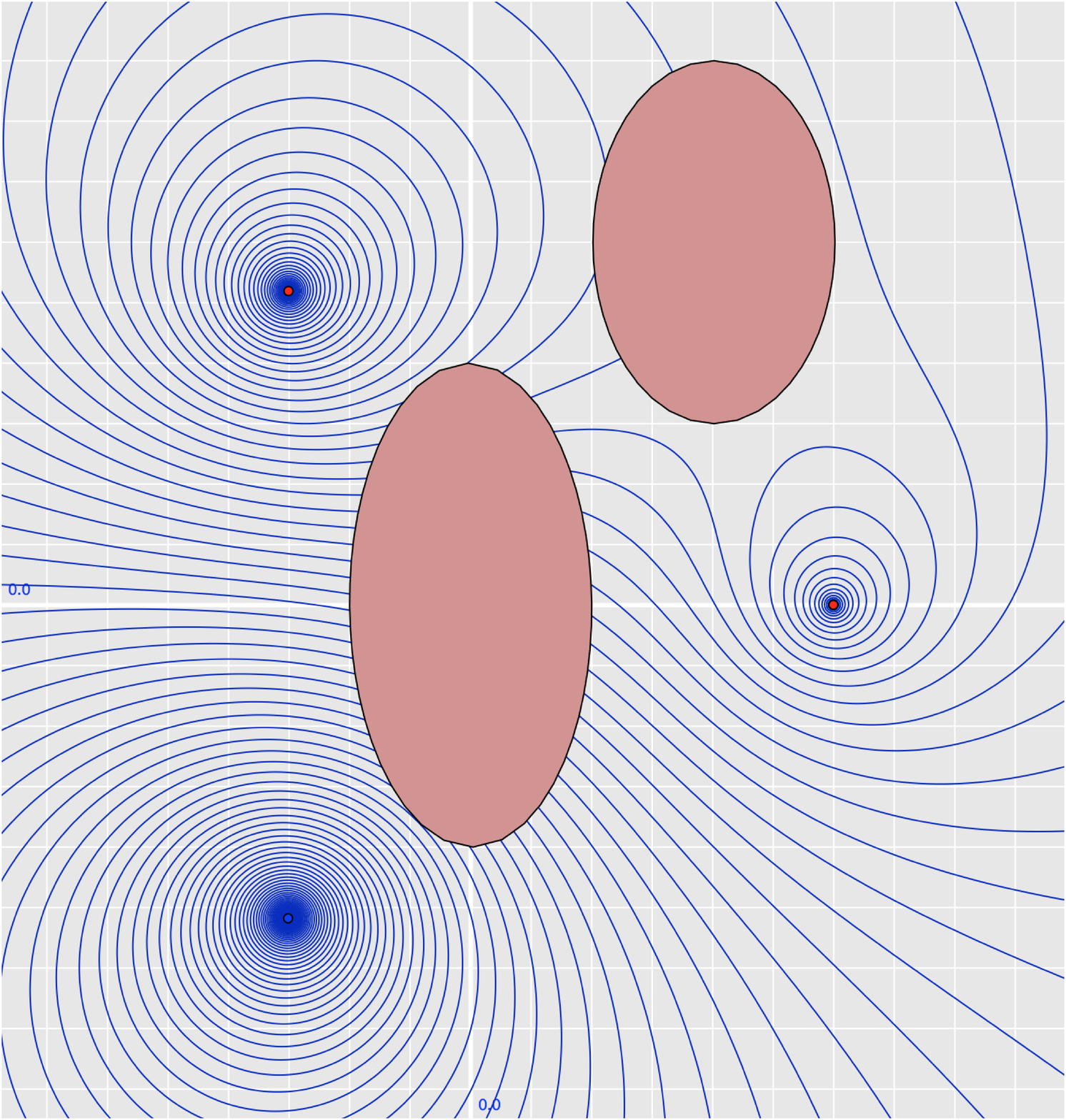}\,\,\,
\includegraphics[width=0.4\columnwidth]{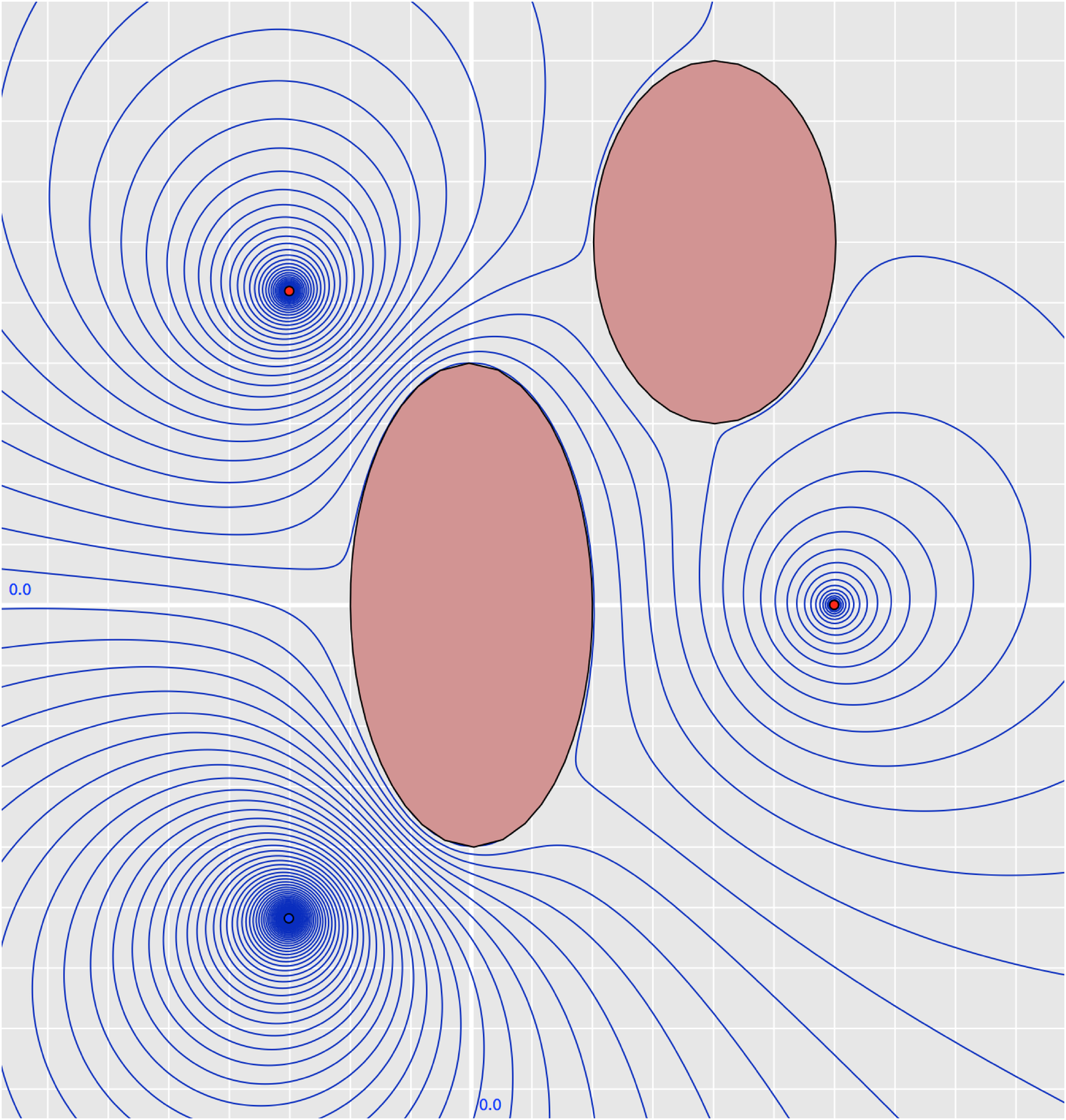}
\caption{Two bodies surrounded by three point vortices with strengths
$1$, $2$, $-3$. The instantaneous fluid motion is along the stream lines (blue),
with velocity proportional to the level density.
\emph{Left:} Stream lines of the velocity field $u_\gamma$ generated by the point
vortices, ignoring the bodies.
\emph{Right:} Stream lines of the velocity field $u_\| = u_\gamma + u_I$. Here the image
vorticity makes the fluid flow nicely around the bodies.}
\label{fig:image_vorticity}
\end{figure}

We will now construct the fluid velocity field $u$ for a given configuration
$(g, \xi, \gamma)$ of $m$ bodies and $n$ point vortices. Here, $g$ and $\xi$
contain the individual body poses $g_j$ and motion states $\xi_j$, and
$\gamma$ encodes the $m$ point vortex positions $\gamma_i$. In the
absence of bodies, the fluid velocity field whose vorticity is given by the
point vortices $\gamma$ with strengths $K$ is determined by the Biot-Savart law:
\begin{equation}
\label{eqn:biot-savart}
u_\gamma(z) = \sum_i K_i \left( J \frac{\gamma_i - z}{\norm{\gamma_i - z}^2}
\right), \qquad J = \mat{0 & -1\\1 & 0}.
\end{equation}
When bodies are present, the velocity field $u_\gamma$ makes fluid particles
move across the body boundaries, see Figure~\ref{fig:image_vorticity}, left.
To fix this, we construct a potential field $u_I = \grad \phi_I$ on $\Fld$
which compensates the normal flux of $u_\gamma$ through the body boundaries,
thus satisfying the boundary condition
\begin{align}
\label{eqn:flux-boundary-condition}
\dd{u_I(z), n_j(z)} & = - \dd{u_\gamma(z), n_j(z)} , \text{ for } z \in \partial {\Bdy}_j.
\end{align}
The subscript $I$ reflects the fact that $u_I$ can be represented as \emph{image
vorticity} inside of the bodies or on their boundaries, see
\citet[\S 2.4]{Saffman1992}. The potential $\phi_I$ of $u_I$ is
uniquely determined by the Neumann problem
\begin{equation}
\label{eqn:image-neumann-problem}
\begin{aligned}
\pard{\phi_I}{n}(z) & = - \dd{u_\gamma(z), n(z)},  \text{ for } z \in
\partial \Fld, & \Delta \phi_I(z) & = 0, 
& \lim_{z \rightarrow \infty} u_I(z) & = 0.
\end{aligned}
\end{equation}
The superposition $u_\| = u_\gamma + u_I$ satisfies the boundary condition
$\tdd{u_\|(z), n(z)} = 0$ on $\partial \Fld$, see Figure~\ref{fig:image_vorticity}, right. In
other words, it is the correct fluid velocity field as long as the bodies are at rest.

When the bodies move we achieve boundary
condition~\eqref{eqn:motion-boundary-condition} by adding another potential
field $u_{\Bdy} = \grad \phi_{\Bdy}$, obtained from the Neumann problem
\begin{equation}
\label{eqn:motion-neumann-problem}
\begin{aligned}
\pard{\phi_\Bdy}{n}(z) & = \dd{
\omega_j \times z + v_j
, n_j(z)}, \text{ for } z \in \partial \Bdy_j, & \Delta
\phi_\Bdy(z) & = 0, & \lim_{z \rightarrow \infty} u_\Bdy(z) & = 0.
\end{aligned}
\end{equation}
The superposition $u = u_\gamma + u_I + u_\Bdy$ is the unique fluid velocity
field which satisfies boundary condition~\eqref{eqn:motion-boundary-condition},
has zero circulation around the individual bodies, vanishes at infinity, and
has its' vorticity field is given by the point vortices $\gamma_i$ with
strengths $K_i$.

The velocity potential $\phi_\Bdy$ depends linearly on body velocities
$\xi_j = (\omega_j, v_j)$, due to the linearity of the Neumann problem. Because
of~\eqref{eqn:adjoint-operator} it also depends linearly on $\Xi_j = (\Omega_j,
V_j)$, i.e., on velocity in the body frame. We use the notation $\Phi_\Bdy$
and $\Phi_{\Bdy_j}$ for the corresponding vector-valued potential in body
coordinates:
\begin{align*}
\Phi_{\Bdy_j}(z) & = (\phi_\Bdy^{3 j}(z), \phi_\Bdy^{3 j + 1}(z), \phi_\Bdy^{3 j
+ 2}(z)) \in \R^3, &
\Phi_\Bdy(z) & = (\phi_\Bdy^1(z), \ldots, \phi_\Bdy^{3 m}(z)) \in \R^{3 m}.
\end{align*}
Then we can write $\phi_\Bdy$, using the standard inner product, as
\begin{align}
\label{eqn:linear-body-phi}
\phi_\Bdy(z) & = \sum_j \dd{\Phi_{\Bdy_j}(z), \Xi_j} = \dd{\Phi_\Bdy(z), \Xi}.
\end{align}
Equivalently we can represent the velocity fields $u_I$ and $u_\Bdy$ in terms of
their stream functions $\psi_I$ and $\psi_\Bdy$. That is:
\begin{align*}
u = \grad \phi = J \grad \psi.
\end{align*}
As for the potential $\phi_\Bdy$, also $\psi_\Bdy$ depends linearly
on body velocity. In analogy to~\eqref{eqn:linear-body-phi} we denote the
vector-valued stream functions by $\Psi_\Bdy$ and $\Psi_{\Bdy_j}$:
\begin{align}
\label{eqn:linear-body-psi}
\psi_\Bdy(z) & = \sum_j \dd{\Psi_{\Bdy_j}(z), \Xi_j} = \dd{\Psi_\Bdy(z), \Xi}.
\end{align}
This representation is important since kinetic energy and the equations of
motion are most easily expressed using stream functions.

\subsection{Kinetic Energy}

The kinetic energy of the fluid velocity field $u$ is
\begin{equation*}
\frac12 \int_{\Fld} \norm{u}^2 = \frac12 \int_{\Fld} \norm{u_\|}^2 + \frac12 \int_{\Fld} \norm{u_{\Bdy}}^2,
\end{equation*}
since $u_\|$ and $u_{\Bdy}$ are $L^2$-orthogonal. The kinetic energy associated
to $u_{\Bdy}$ can be expressed in terms of \emph{added mass} as a quadratic form
on body velocity, similar to the kinetic
energy~\eqref{eqn:rigid-body-kinetic-energy}:
\begin{align}
\label{eqn:kirchhoff-energy}
\frac12 \int_{\Fld} \norm{u_{\Bdy}}^2 = \frac12 \dd{\K(g) \, \Xi, \Xi}.
\end{align}
The matrix $\K(g)$ is called the \emph{added mass tensor}, and the main
difference to the mass-inertia tensor $\MI$ is that $\K(g)$ depends on $g$,
while $\MI$ is a constant matrix which depends only on the body shapes and mass
distributions. For a single rigid body the derivation of
\eqref{eqn:kirchhoff-energy} goes back to \citet{Kirchhoff1870}, see
\citet{Nair2007} for the case of several rigid bodies.
Together with the kinetic energy associated to the rigid body motion, we define
the \emph{Kirchhoff tensor} $\KT = \MI+\K(g)$. The corresponding energy,
\begin{align}
\label{eqn:kinetic-energy-hydrodynamically-coupled}
T_\Bdy = \frac12 \dd{\KT \, \Xi, \Xi},
\end{align}
can be viewed as a function of Kirchhoff momentum $M = \KT \Xi$. Then
\begin{align}
\label{eqn:hydrodynamically-coupled-hamiltonian}
H_\Bdy = \frac12 \dd{M, \KT^{-1} M}
\end{align}
is the Hamiltonian for rigid bodies in potential flow.

The kinetic energy associated to $u_\|$ contains infinite self-energy terms that
we ignore.\footnote{This is a special feature of the 2D case, where a single
point vortex in an unbounded fluid domain will not move at all due to symmetry.
In 3D, when considering vortex filaments, the self-energy has significant influence on
the dynamics and cannot be ignored.} The finite part is given as the negative of the
\emph{Kirchhoff-Routh function} $W_G$ \citep{Lin1941,Shashikanth2005}:
\begin{equation}
\label{eqn:kirchhoff-routh-function}
W_G = \frac12 \sum_i K_i \left( \psi_I(\gamma_i) + \sum_{i \neq j} \psi_{\gamma_j}(\gamma_i) \right),
\end{equation}
where $\psi_{\gamma_j}(z) = - K_j \log{\norm{\gamma_j - z}}$ is the stream function
of the point vortex $\gamma_j$. In the absence of bodies (or other boundaries)
the kinetic energy of $u_\| = u_\gamma$ reduces to
\begin{equation}
\label{eqn:point-vortex-hamiltonian}
H_\gamma = - \frac12 \sum_i K_i \sum_{j \neq i} \psi_{\gamma_j}(\gamma_i),
\end{equation}
which is the Hamiltonian of $m$ point vortices in the plane. Note that gradient
of $W_G$ with respect to a point vortex $\gamma_i$ encodes the velocity field
$u_\|(\gamma_j)$:
\begin{align}
\label{eqn:energy-variation-gamma}
\grad_{\gamma_i} W_G = - K_i J u_\|(\gamma_i).
\end{align}
The total kinetic energy of the coupled system is
\begin{equation}
\label{eqn:kinetic-energy}
H = H_\Bdy - W_G = T_\Bdy - W_G.
\end{equation}

\section{Equations of Motion}

We consider $m$ rigid bodies (topological disks) in the plane, surrounded by an
inviscid and incompressible fluid. The circulation around the
individual bodies is zero, and the whole vorticity of the fluid is concentrated
at $n$ isolated point vortices $\gamma_i \in \R^2$, with strengths
$K_i$. Euclidean transformations and velocity states of the
different bodies are denoted by $g_j$ and $\Xi_j = (\Omega_j, V_j)$, as in
Section~\ref{sec:rigid-bodies}. We define\footnote{The physical meaning of these
quantities will be discussed in Sections~\ref{sec:point-vortex-dynamics}
and~\ref{sec:coupled-dynamics}.}
\begin{align*}
M_{C} & = \sum_i K_i \Psi_\Bdy(\gamma_i), & W_\Bdy & = \sum_i K_i
\psi_\Bdy(\gamma_i) = \dd{M_C, \Xi},
\end{align*}
and the generalized momentum $M = \KT \, \Xi + M_C$ with components $M_j =
(A_j, L_j)$ corresponding to angular and linear momentum.\\

\begin{theorem}
\label{thm:equations-of-motion}
The motion of the coupled system described above is governed by the following
system of differential equations:
\begin{equation}
\label{eqn:equations-of-motion}
\begin{aligned}
\vec{\dot{A}_j + V_j \times L_j \\[0.05cm] \dot{L}_j + \Omega_j \times
L_j} & = - \Big( \lgrad (H_\Bdy - W_G - W_\Bdy) \Big)_j,\\[0.1cm]
\Xi & = \KT^{-1} (M - M_C), \\[0.1cm]
 \dot{g}_j & = g_j
\Xi_j,\\[0.1cm] \dot{\gamma}_i & = u(\gamma_i).
\end{aligned}
\end{equation}
\end{theorem}
This theorem will be proven in the remainder of this section.

\subsection{Hamiltonian Formulation}

The starting point for our derivation was made by \cite{Vankerschaver2009}, who
determined the dynamics of a single, circular disk with point vortices through
the framework of cotangent bundle reduction \citep{Marsden2007}.
Appendix~\ref{app:cotangent-bundle-reduction} summarizes the derivation of the
reduced phase space for the system. We emphasize here that the derivation
readily generalizes to several bodies. Apart from the reduced phase space, one
needs to determine the reduced symplectic form in order to fully describe the
dynamics. \cite{Vankerschaver2009} use the general framework of cotangent bundle
reduction, which determines the symplectic form through the curvature of the
mechanical connection. While it is in principle possible to extend this approach
to several bodies of arbitrary shape, the derivation requires heavy
differential-geometric machinery, takes place in unreduced, infinite-dimensional
phase space, and is already challenging for the case of a single, circular disk. 

We make use of the fact that the dynamics of the surrounding fluid,
and thus the point vortex dynamics, is completely determined from first
principles. This allows to determine the symplectic form in finite-dimensional
reduced phase space, using standard methods. In particular, we use the fact
that point vortices are, by Helmholtz' law, advected along the fluid velocity
field $u$. Further we will show that the dynamics of rigid bodies and point
vortices decouples asymptotically. These two properties uniquely determine the
symplectic form, and we obtain the equations of motion by computing Hamilton's
equations
\begin{align}
\label{eqn:hamilton-equations}
\sigma(X, \dot{q}) = \d H(X), \quad \forall_X \in T_q {\mathcal M}
.
\end{align}
Here $\mathcal{M}$ is the reduced phase space, $\sigma$ is the symplectic form,
and $H \colon \, {\mathcal M} \rightarrow \R$ is the Hamiltonian. The reduced
phase space for the coupled system of $n$ bodies and $m$ point vortices is (see
Appendix~\ref{app:cotangent-bundle-reduction})
\begin{equation}
{\mathcal M} = T^* \SE(2)^n \times \R^{2 m},
\end{equation}
where $T^*\SE(2)^n$ is the cotangent bundle of $\SE(2)^n$. A point $q=(\mu, g,
\gamma) \in {\mathcal M}$ encodes body poses $g \in SE(2)^n$,
body momentum through the covector $\mu \in T_g^*SE(2)^n$, and the $m$ point
vortex locations $\gamma \in \R^{2 m}$. The two factors of ${\mathcal M}$ are
already symplectic manifolds, they are the phase spaces of the two
uncoupled systems:
\begin{itemize}
  \item $(T^* \SE(2)^n, \sigma_{can})$ is the phase space of $n$ rigid
  bodies. Any cotangent bundle carries a canonical symplectic structure,
  in this case $\sigma_{can} = \d \mu \wedge \d g$. 
  Hamilton's equations with Hamiltonian
  $H_\Bdy$~\eqref{eqn:hydrodynamically-coupled-hamiltonian} describe the motion
  of $n$ rigid bodies in potential flow.
  \item $(\R^{2 m}, \sigma_\gamma)$ is the phase space of $m$ point vortices in
  the plane. The symplectic form is $\sigma_\gamma = - \sum_i K_i \d x_i \wedge
  \d y_i$, the weighted sum of canonical symplectic forms on the individual
  $\R^2$ factors. Hamilton's equations 
  with Hamiltonian $H_\gamma$~\eqref{eqn:point-vortex-hamiltonian} describe the
  motion of $m$ point vortices in the plane.
\end{itemize}

We know from general theory of cotangent bundle reduction that the
Hamiltonian of the system is the kinetic energy~\eqref{eqn:kinetic-energy},
and that the symplectic form $\sigma$ on ${\mathcal M}$ is of the form
\begin{equation}
\sigma = \sigma_{can} + \d \alpha + \sigma_\gamma.
\end{equation}
Here $\sigma_{can}$ and $\sigma_\gamma$ are the symplectic forms on the
individual factors, and $\d \alpha$ is a \emph{magnetic} or \emph{Coriolis}
term, which is responsible for the dynamical coupling between rigid bodies and
point vortices. The two-form $\d \alpha$ lives on $\SE(2)^n \times \R^{2 m}$,
i.e., it is independent of $\mu$. In the remainder of this section we will
determine the magnetic term $\d \alpha$, and then derive the equations of
motion~\eqref{eqn:equations-of-motion} from Hamilton's
equations~\eqref{eqn:hamilton-equations}.

\subsubsection{Asymptotic Decoupling}
\label{sec:asymptotic}

We will now study the behavior of the coupled dynamical system in the limit
when rigid bodies and point vortices are far apart.
Assume that both bodies and point vortices are contained in two disjoint disks
of finite radius, and let $d$ denote the minimal distance between these disks.
We will now show that in the limit $d \rightarrow \infty$ the system
decouples:\\

\begin{lemma}
\label{lemma:decoupling}
The dynamical system decouples in the limit $d \rightarrow \infty$, i.e.,
\begin{equation*}
\lim_{d \rightarrow \infty} \d \alpha = 0.
\end{equation*}
\end{lemma}

\begin{proof}
We consider the difference of Hamilton's equations of the coupled and the two
uncoupled systems. For the difference of the symplectic forms we obtain
\begin{align*}
\sigma - \sigma_{can} - \sigma_\gamma = \d \alpha,
\end{align*}
and the difference of the Hamiltonians is
\begin{align*}
H - H_\Bdy - H_\gamma =  - W_G - H_\gamma = -\frac12
\sum_i K_i \psi_I(\gamma_i).
\end{align*}
The difference of Hamilton's equations in the limit $d \rightarrow \infty$
is
\begin{align*}
\lim_{d \rightarrow \infty} \d \alpha((\Gamma,
\delta{\gamma}), (\Xi, \dot{\gamma})) & = \lim_{d \rightarrow \infty} \d (H -
H_\Bdy - H_\gamma)(\Gamma, \delta{\gamma})\\
& = \lim_{d \rightarrow \infty} \sum_i K_i \left( - \frac12
\delta_g \psi_I(\gamma_i)
 + \tdd{J u_I(\gamma_i), {\delta{\gamma}}_i} \right).
\end{align*}
The right hand side vanishes because of
Lemma~\ref{lemma:falloff-stream-functions} and the construction of $u_I$, see
the Equation~\eqref{eqn:image-neumann-problem}.
\end{proof}

It remains to determine the asymptotic behavior of the stream functions
$\psi_I$, $\psi_\Bdy$ and their variations with respect to $g$:\\

\begin{lemma}
\label{lemma:falloff-stream-functions}
The functions $\psi_I$, $\psi_\Bdy$, $\delta_g \psi_I$ and $\delta_g
\psi_\Bdy$ are ${\mathcal O}(|z|^{-1})$.
\end{lemma}
\begin{proof}
Let $u$ be $u_I$ or $u_\Bdy$, and $\psi$ the stream function of $u$. We represent
$\psi$ as a single layer potential with density $\tau$, i.e.,
\begin{align*}
\psi(z) = & \sum_j\nolimits
\psi^j(z) & \psi^j(z) = & \oint_{\partial \Bdy_j} \tau(\eta) \log
\norm{\eta-g_j^{-1}(z)} \, d\eta.
\end{align*}
Note that we can also view $\tau$ as the strength of a vortex sheet (i.e., a distribution
of point vortices) on $\partial \Bdy$, which generates $u$
via the Biot-Savart law~\eqref{eqn:biot-savart}. Since the circulation of $u = \grad \phi$ around
any $\Bdy_j$ is zero, the total density of $\tau$ on each
$\partial \Bdy_j$, as well as its variation with respect to $g$, vanish:
\begin{equation}
\label{eqn:stream-density-zero}
\oint_{\partial \Bdy_j} \tau(\eta) \, d\eta = \oint_{\partial \Bdy_j}
(\delta_g \tau)(\eta) \, d\eta = 0.
\end{equation}
This implies $\psi^j(z) = \psi(z) = {\mathcal O}(|z|^{-1})$. For the variation with
respect to $g$ we obtain
\begin{align*}
\delta_g \psi^j(z) = 
    \dd{\oint_{\partial \Bdy_j} \tau(\eta)
    \frac{\eta-g_j^{-1}(z)}{\norm{\eta-g_j^{-1}(z)}^2} \, d\eta,  -
    \delta g_j^{-1}(z) } + \oint_{\partial \Bdy_j} (\delta_g \tau)(\eta)
    \log \norm{\eta - g_j^{-1}(z)} \, d\eta.
\end{align*}
The second term is also a single layer potential and ${\mathcal O}(|z|^{-1})$
because of~\eqref{eqn:stream-density-zero}. The boundary integral in the first
term is a rotated version (by $-\pi/2$) of the velocity field $u$
generated by the single layer density $\tau$ on $\partial \Bdy_j$. Thus it is ${\mathcal O}(|z|^{-2})$
because of~\eqref{eqn:stream-density-zero}, while $\delta g_j^{-1}(z)$ is
${\mathcal O}(|z|)$. So the first term is ${\mathcal O}(|z|^{-1})$ as well.
\end{proof}

\subsubsection{Dynamics of Hydrodynamically Coupled Rigid Bodies}

The dynamics of rigid bodies in potential flow is a canonical Hamiltonian system
with phase space ${\mathcal M} = T^* \SE(2)^n$ and Hamiltonian $H_\Bdy$, see
Equation~\eqref{eqn:hydrodynamically-coupled-hamiltonian}.
We use the left trivialization
\begin{equation}
\mathcal{M} = \R^{3 n} \times SE(2)^n \cong T^* \SE(2)^n
\end{equation}
that is, we identify a covector $\mu \in T_g^*\SE(2)^n$ with
its corresponding body momentum $M \in \R^{3 n}$. In this way we
obtain the equations of motion in the body frame of reference, as an
evolution equation for $M$.
The canonical symplectic form on $\mathcal{M}$ is derived in
Appendix~\ref{app:cotangent-bundle}. Now we compute Hamilton's
equations~\eqref{eqn:hamilton-equations} using the Hamiltonian
$H_\Bdy$~\eqref{eqn:hydrodynamically-coupled-hamiltonian}. We have $\dot{q} =
(\dot{M}, \Xi)$, $X = (\delta M, \Gamma)$, and obtain
\begin{gather*}
\dd{\delta{M}, \Xi} - \dd{\dot{M} - \coad{\Xi} M, \Gamma} =
\dd{\delta{M}, \KT^{-1} M} + \dd{\lgrad H_\Bdy, \Gamma}.
\end{gather*}
Comparing $\delta{M}$-coefficients gives $M = \KT \Xi$ while the
$\Gamma$-coefficients give the equations of motion as an evolution equation for
$M$:
\begin{align*}
\dot{M} - \coad{\Xi} M = - \lgrad H_\Bdy \qquad \Longleftrightarrow 
\qquad
\vec{\dot{A}_j + V_j
\times L_j
\\
\dot{L}_j + \Omega_j \times L_j} = - \big( \lgrad H_\Bdy \big)_j.
\end{align*}

\subsubsection{Point Vortex Dynamics}
\label{sec:point-vortex-dynamics}

According to Helmholtz' law (see, for instance, \citet{Saffman1992}), point
vortices are frozen into the fluid velocity field:
$\dot{\gamma}_i = u(\gamma_i)$. The Hamiltonian formulation of the system is
obtained as follows. The symplectic form for $m$ point vortices $\gamma \in
\R^{2 m}$ with strengths $K_i$ is
\begin{align*}
\sigma_{\gamma} = - \sum_i K_i \, \d x_i \wedge \d y_i, \quad \text{ i.e., }
\quad \sigma_\gamma(\delta{\gamma}, \dot{\gamma})
= \sum_i K_i \dd{\delta{\gamma}_i, J \dot{\gamma}_i}.
\end{align*}
In the absence of boundaries the Hamiltonian is
$H_\gamma$~\eqref{eqn:point-vortex-hamiltonian}. Computing Hamilton's equations,
\begin{equation*}
\sigma_\gamma(\delta{\gamma}, \dot{\gamma}) = \d H_\gamma(\delta{\gamma})
\qquad \Longleftrightarrow \qquad
\sum_i K_i \dd{{\delta{\gamma}}_i, J \dot{\gamma}_i} = \sum_i K_i
\dd{J u_\gamma(\gamma_i),  {\delta{\gamma}}_i},
\end{equation*}
we obtain, as expected, $\dot{\gamma}_i = u_\gamma(\gamma_i)$. This holds also
for a fluid with fixed boundaries, i.e., for fixed body configuration. In this
case we can choose the negative of the Kirchhoff-Routh function $-W_G$
(Equation~\eqref{eqn:kirchhoff-routh-function}) as the Hamiltonian and obtain
$\dot{\gamma}_i = u_\|(\gamma_i)$. In both cases the Hamiltonian is the kinetic
energy of the fluid velocity field (with infinite self-energy terms excluded),
and the Hamiltonian is a constant of motion.

In the case of rigidly moving boundaries we use the \emph{generalized
Kirchhoff-Routh function}, extended to rigidly moving boundaries by
\citet{Shashikanth2002}:
\begin{equation}
\label{eqn:generalized-kirchhoff-routh-function}
W = W_G + W_\Bdy, \quad W_\Bdy = \sum_i K_i \psi_\Bdy(\gamma_i).
\end{equation}
Choosing $-W$ as the Hamiltonian gives the correct point vortex dynamics
for the case that some agency moves the bodies around. Here the
Hamiltonian depends explicitly on time, thus it is not a constant of motion.

\subsubsection{Dynamics of the Coupled System}
\label{sec:coupled-dynamics}

Looking at the dynamics of point vortices in a fluid with rigidly moving
boundaries (Section~\ref{sec:point-vortex-dynamics}), it is not surprising that
the function $W_\Bdy$~\eqref{eqn:generalized-kirchhoff-routh-function} needs to
make its way into Hamilton's equations of the coupled system. However, the
Hamiltonian of the system is already known: It is the kinetic energy $H=-W_G +
H_\Bdy$ of the combined fluid-body system, given in
Equation~\eqref{eqn:kinetic-energy}. On the other hand, $W_\Bdy$ can be viewed
as a one-form on $\SE(2)^n \times \R^{2 m}$. We will now show that $\alpha =
W_\Bdy$ is in fact a primitive of the magnetic term $\d \alpha$.\\

\begin{lemma}
\label{lemma:magnetic-symplectic-form}
The one-form
\begin{align}
\alpha(\Xi, \dot{\gamma}) & := W_\Bdy = \dd{M_C, \Xi}, & 
M_C & = \sum_i K_i \Psi_\Bdy(\gamma_i)
\end{align}
on $\SE(2)^n \times \R^{2 m}$ is a primitive one-form of the magnetic term $\d
\alpha$.
\end{lemma}

\begin{proof}
Let us compute Hamilton's equations for the coupled system, i.e.,
\begin{align}
\label{eqn:hamilton-equations-coupled-1}
\sigma\left(X, \dot{q}\right) = \d H (X),
\end{align}
with  $\dot{q} = (\dot{M}, \Xi, \dot{\gamma})$ and $X = (\delta{M}, \Gamma,
\delta{\gamma})$ is an arbitrary tangent vector at $q=(M, g, \gamma)$:
\begin{equation}
\label{eqn:hamilton-equations-coupled-2}
\begin{aligned}
\d H(X)
& = \dd{\delta{M}, \KT^{-1} M} + \dd{\lgrad_g H, \Gamma }
+ \sum_i K_i \dd{J u_\|(\gamma_i), {\delta{\gamma}}_i},\\
\sigma(X, \dot{q}) 
 & = \dd{\delta{M}, \Xi} - \dd{\dot{M} - \coad{\Xi} M,
\Gamma}
+ \sigma_\gamma(X, \dot{\gamma})
+ \d \alpha \left((\Gamma, \delta{\gamma}), (\Xi, \dot{\gamma}) \right).
\end{aligned}
\end{equation}
From the $\delta{M}$-coefficients we immediately obtain $M = \KT \Xi$. Further,
using the linearity of $\d \alpha$ with respect to $(\Gamma, \delta{\gamma})$,
we have
\begin{align}
\label{eqn:variation-gamma}
\sigma_\gamma(\delta{\gamma}, \dot{\gamma})
+ \d \alpha((0, \delta{\gamma}), (\Xi, \dot{\gamma})) = \sum_i K_i
\dd{J u_\|(\gamma_i), {\delta{\gamma}}_i}.
\end{align}
By virtue of Helmholtz' law point vortices are frozen into the fluid, i.e.,
$\dot{\gamma}_i = u(\gamma_i)$. Therefore we can rewrite the right hand side
using $u_\| = u - u_\Bdy$ and obtain
\begin{align*}
\sum_i K_i \dd{J u_\|(\gamma_i), {\delta{\gamma}}_i} & = 
\sum_i K_i \dd{J (\dot{\gamma}_i - u_\Bdy(\gamma_i)),
{\delta{\gamma}}_i}\\
 & = \sigma_\gamma(\delta{\gamma}, \dot{\gamma}) + \dd{\grad_\gamma W_\Bdy,
\delta{\gamma}}
= \sigma_\gamma(\delta{\gamma}, \dot{\gamma}) + \dd{\grad_\gamma \tdd{M_C,
\Xi}, \delta{\gamma}}.
\end{align*}
Subtracting $\sigma_\gamma$ on both sides of~\eqref{eqn:variation-gamma} we have
\begin{align}
\label{eqn:dAlpha1}
\d \alpha \left((0, \delta{\gamma}), (\Xi, \dot{\gamma}) \right) =
\dd{\grad_\gamma \tdd{M_C, \Xi}, \delta{\gamma}}.
\end{align}
From the general theory of cotangent-bundle reduction we know that $\d \alpha$
does not dependent on $M$. In particular, we can consider $M=0$ which implies $\Xi = 0$
and thus verify that $\d \alpha \left((0, \delta{\gamma}), (0, \dot{\gamma})
\right) = 0$. As a consequence, we can assume $\alpha$ to be of the
form\footnote{Any one-form $\Theta$ on $\SE(2)^n \times \R^{2 m}$ can be written
as $\Theta(\Xi, \dot{\gamma}) = \tdd{M_g(g, \gamma), \Xi} + \tdd{M_\gamma(g,
\gamma), \dot{\gamma}}$.}
\begin{equation*}
\alpha(\Xi, \dot{\gamma}) = \dd{M_\alpha(g, \gamma), \Xi}.
\end{equation*}
We compute\footnote{As in the derivation of $\sigma_{can}$ in
Appendix~\ref{app:cotangent-bundle}, we use Equation~\eqref{eqn:d-omega} and
consider a 2--parameter family $g(s, t)$ with commuting partial derivatives.
This makes the Lie bracket term in~\eqref{eqn:d-omega} vanish, but introduces the
constraint~\eqref{eqn:velocity-integrability} on $\dot{\Gamma}$.} $\d \alpha$ and obtain
\begin{equation}
\label{eqn:dAlpha2}
\begin{aligned}
\d \alpha((\Gamma, \delta{\gamma}), (\Xi, \dot{\gamma}))
& = \delta \dd{{M}_\alpha, \Xi} - \dd{{M}_\alpha, \Gamma}\dot{}\\
& = \dd{\lgrad \tdd{M_\alpha, \Xi}, \Gamma}
 + \dd{\grad_\gamma \tdd{M_\alpha, \Xi}, \delta \gamma}
 - \dd{\dot{M}_\alpha - \coad{\Xi} M_\alpha, \Gamma}.
\end{aligned}
\end{equation}
When equating \eqref{eqn:dAlpha1} and \eqref{eqn:dAlpha2} (for $\Gamma=0$) we
obtain
\begin{align*}
\grad_\gamma \dd{M_\alpha, \Xi} = \grad_\gamma \dd{M_C, \Xi}.
\end{align*}
This determines $M_\alpha$ up to a contribution that is independent of
$\gamma$, i.e., $M_\alpha = M_C + \tilde{M}(g)$. It remains to be shown that
$\tilde{M}$ does not contribute to $\d \alpha$. Let us write $\alpha =
\hat{\alpha} + \tilde{\alpha}$ with
\begin{align*}
\hat{\alpha}(\Xi, \dot{\gamma}) & = \dd{M_C, \Xi}, &
\tilde{\alpha}(\Xi, \dot{\gamma}) & = \dd{\tilde{M}, \Xi}.
\end{align*}
When we consider the limit $d \rightarrow \infty$ (distance between point
vortices and bodies), it follows from Lemma~\ref{lemma:falloff-stream-functions}
that $\lim_{d \rightarrow \infty} \d \hat{\alpha} = 0$.
Lemma~\ref{lemma:decoupling} on the other hand guarantees that the system
decouples in this limit, i.e., $\lim_{d \rightarrow \infty} \d \alpha = 0$. It
follows that $\d \tilde{\alpha}$ vanishes identically, since it is independent
of $\gamma$, and thus of $d$.
\end{proof}
Now we are able to verify the equations of motion given
in Theorem~\ref{thm:equations-of-motion}:
\begin{proof}[Proof of Theorem~\ref{thm:equations-of-motion}]
The magnetic term is (see
Equation~\eqref{eqn:dAlpha2}, with $M_\alpha = M_C$)
\begin{align*}
\d \alpha & = \dd{ \lgrad W_\Bdy, \Gamma} - \sum_i K_i \dd{J u_\Bdy(\gamma_i), \delta
\gamma_i} - \dd{\dot{M}_C - \coad{\Xi} M_C, \Gamma}.
\end{align*}
Substituting into Hamilton's equations~(\eqref{eqn:hamilton-equations-coupled-1}
and~\eqref{eqn:hamilton-equations-coupled-2}) and bringing the first two terms
to the right hand side gives
\begin{align*}
RHS & = \dd{\delta{M}, \KT^{-1} M} + \dd{\lgrad_g (H_\Bdy - W), \Gamma }
+ \sum_i K_i \dd{J u(\gamma_i), {\delta{\gamma}}_i}, \\
LHS & = \dd{\delta M, \Xi} - \dd{(\dot{M}+\dot{M}_C) - \coad{\Xi} (M+M_C),
\Gamma} + \sum_i K_i \dd{J \dot{\gamma}_i, \delta{\gamma}_i}.
\end{align*}
Comparing coefficients gives the equations of
motion~\eqref{eqn:equations-of-motion}.
\end{proof}

\section{Lagrangian Formulation and Total Momentum}
\label{sec:lagrangian}

For any canonical Hamiltonian system on a cotangent bundle $\mathcal{M} = T^*
\mathcal{Q}$ with canonical symplectic form $\sigma = \d \Theta$, a Lagrangian
description is obtained through the Legendre transformation. That is, momentum
is expressed as a function on the tangent bundle $T \mathcal{Q}$, and the
Lagrangian is $L = \Theta - H$, also viewed as a function on $T \mathcal{Q}$.
The corresponding Euler-Lagrange equations give the same dynamics as the
Hamiltonian system, a classical result which can be found in any mechanics
textbook.

The above construction does not apply to point vortices, since the phase space
is not a cotangent bundle, and it is thus not clear (and in general not even
possible) how to split this space into configurations and momenta. Nevertheless,
if $\beta_\gamma$ is a primitive for $\sigma_\gamma = \d \beta_\gamma$, one can
show that the Lagrangian $L = \beta_\gamma - H_\gamma$ describes the dynamics of
point vortices, see \citet{Rowley2002}. We will now verify that this approach
also gives a Lagrangian description for the coupled dynamics of rigid bodies and
point vortices.

A primitive one-form of $\sigma = \sigma_{can} + \sigma_\gamma + \d \alpha$ is
easily found: The canonical symplectic form on $T^* \SE(2)^n$ is the exterior
derivative of $\beta_{can} = \tdd{M, \Xi} =  \tdd{\KT \Xi, \Xi}$,
see Appendix~\ref{app:cotangent-bundle}. For $\sigma_\gamma$ we use the
primitive
\begin{align*}
\beta_\gamma = - \frac{1}{2} \sum_i K_i \det(\gamma_i,
\dot{\gamma}_i) = - \frac{1}{2} \sum_i K_i \dd{\gamma_i, J \dot{\gamma}_i}.
\end{align*}
With $\alpha = W_\Bdy = \tdd{M_C, \Xi}$ obtain as the Lagrangian
\begin{align}
\label{eqn:lagrangian}
L & = \beta - H = \underbrace{\dd{\KT \Xi, \Xi}}_{2 T_\Bdy} + \beta_\gamma
+W_\Bdy - T_\Bdy + W_G = T_\Bdy + \beta_\gamma + W,
\end{align}
where we have expressed the Kirchhoff kinetic energy $H_\Bdy$ as a function on
the tangent bundle, denoted by $T_\Bdy$. We will now verify that $L$ is indeed a
Lagrangian for the system. At the same time we will determine the total momentum
of the system, by keeping track of the end points when using integration by parts:
\begin{align*}
0 = \delta S_L & = \delta \int_{t_0}^{t_1} L \, dt = \int_{t_0}^{t_1} \delta
\beta_\gamma + \delta (T_\Bdy + W) \, dt\\
& = \int_{t_0}^{t_1} \dd{\lgrad (T_\Bdy + W), \Gamma} + \dd{\underbrace{\KT \Xi
+ M_C}_{M}, \delta \Xi} \, dt\\
& + \sum_i K_i \int_{t_0}^{t_1} \frac12 \dd{\delta \gamma_i, J
\dot{\gamma}_i} + \frac12 \dd{\gamma_i, J \delta \dot{\gamma}_i} -
\dd{J u(\gamma_i), \delta \gamma_i} \, dt.
\end{align*}
Now we use $\delta \Xi = \dot{\Gamma} + \ad{\Xi} \Gamma$ (see
Equation~\eqref{eqn:velocity-integrability} in
Appendix~\ref{app:cotangent-bundle}) and integration by parts, i.e.,
\begin{align*}
\dd{M, \Gamma} \Big|_{t_0}^{t_1} & = \int_{t_0}^{t_1} \dd{\dot{M}, \Gamma} +
\dd{M, \dot{\Gamma}} \, dt,\\
\frac12 \dd{\gamma_i, J \delta{\gamma}_i} \Big|_{t_0}^{t_1}
& = \frac12 \int_{t_0}^{t_1} \dd{\dot{\gamma}_i, J \delta \gamma_i} +
\dd{\gamma_i, \delta \dot{\gamma}_i} \, dt,
\end{align*}
and obtain
\begin{align*}
\delta S_L & = \int_{t_0}^{t_1} \dd{\lgrad(T_\Bdy + W) + \coad{\Xi} M - \dot{M},
\Gamma} + \sum_i K_i \dd{J (\dot{\gamma}_i - u(\gamma_i)), \delta \gamma_i} \,
dt\\
& + \left( \dd{M, \Gamma} + \frac12 \sum_i K_i \dd{\gamma_i, J \delta \gamma_i}
\right) \Big|_{t_0}^{t_1}.
\end{align*}
For variations with fixed end points we obtain the equations of motion
(Theorem~\ref{thm:equations-of-motion}) as critical values of $S_L$. The
Euler-Lagrange equations of the system are
\begin{align}
\label{eqn:euler-lagrange}
\lgrad(T_\Bdy + W) + \coad{\Xi} M - \dot{M} & = 0, & \dot{\gamma}_i - u(\gamma_i) & =
0.
\end{align}
Since $\lgrad T_\Bdy = - \lgrad H_\Bdy$ we have proven:\\

\begin{corollary}
The function $L = T_\Bdy + \beta_\gamma + W$ is a Lagrangian for the coupled
system of rigid bodies and point vortices.
\end{corollary}

The total momentum of the system is obtained by applying an infinitesimal
Euclidean motion to a solution $q$ of the Euler-Lagrange equations. The
corresponding variation $\delta q$ has the form
\begin{align*}
\delta g_j & = \underbrace{\mat{\tilde{\omega} \times & \tilde{v} \\
0 & 0}}_c g_j = g_j \Gamma_j, & \delta \gamma_i = \tilde{\omega} \times \gamma_i
+ \tilde{v}.
\end{align*}
Since $q$ solves the Euler-Lagrange equations~\eqref{eqn:euler-lagrange} we
obtain
\begin{align}
\label{eqn:momentum-1}
\delta S_L = \left( \sum_j \dd{M_j, \Ad{g_j^{-1}} c} + \frac12 \sum_i K_i
\dd{\gamma_i, J \delta \gamma_i} \right) \Big|_{t_0}^{t_1}.
\end{align}
On the other hand, $T_\Bdy + W$ does not change under a Euclidean motion. Hence
\begin{equation}
\label{eqn:momentum-2}
\begin{aligned}
\delta S_L & = \int_{t_0}^{t_1} \delta \beta_\gamma \, dt = \frac12 \sum_i K_i
\int_{t_0}^{t_1} \dd{\delta \gamma_i, J \dot{\gamma}_i} + \dd{\gamma_i, J \delta
\dot{\gamma}_i} \, dt\\
& = \sum_i K_i \int_{t_0}^{t_1} \dd{\delta \gamma_i, J \dot{\gamma}_i} \, dt
+ \frac12 \sum_i K_i \dd{\gamma_i, J \delta{\gamma}_i} \Big|_{t_0}^{t_1}\\
& = \dd{ \sum_i K_i \vec{\frac12 \norm{\gamma_i}^2 \\ J \gamma_i},
\vec{\tilde{\omega} \\ \tilde{v}}}\Big|_{t_0}^{t_1} + \frac12 \sum_i K_i
\dd{\gamma_i, J \delta{\gamma}_i} \Big|_{t_0}^{t_1}.
\end{aligned}
\end{equation}
Here we have again used integration by parts and the fact that
\begin{align*}
\dd{\frac{d}{dt}\Big|_{t=0}\vec{\frac12 \norm{\gamma_i}^2 \\ J \gamma_i},
\vec{\tilde{\omega} \\ \tilde{v}}} = \dd{\delta \gamma_i, J \dot{\gamma}_i}.
\end{align*}
The total momentum is obtained by equating~\eqref{eqn:momentum-1}
and~\eqref{eqn:momentum-2}, and using the fact that this equation holds for any
$t_1$:\\

\begin{corollary}
The coupled system of rigid bodies and point vortices has the following
constants of motion induced by the Euclidean symmetry group:
\begin{align*}
const. & = \sum_j \vec{a_j \\ \ell_j} - \sum_i K_i \vec{\frac12
\norm{\gamma_i}^2
\\
J \gamma_i }, & \vec{a_j \\ \ell_j} & = \CoAd{g_j^{-1}} M_j
 = \vec{A_j  y_j \times R_j L_j \\ R_j L_j}.
\end{align*}
\end{corollary}

\section{Numerical Simulation}
\label{sec:time-integration}

In this section we briefly describe how to implement a numerical method to
simulate the dynamics of the coupled system, and validate our method by
simulating different configurations.

We have chosen to construct a variational integrator~\cite{Marsden2001} for the
system, based on the Lagrangian formulation given in Section~\ref{sec:lagrangian}.
The Lagrangian of the system is partly degenerate, so it fits into the framework
of variational integrators for degenerate Lagrangian systems, see
\citet{Rowley2002}. Some aspects of implementation regarding the Lie group
configuration space can be found in \citep{Kobilarov2009}.

Our implementation uses a midpoint scheme, i.e., we discretize the smooth action
integral by evaluating in between two configurations (along a geodesic
connecting them), and multiplying the corresponding value with the time step:
\begin{align*}
L^d(q_k, q_{k+1}) = h L\left(\evat{(q, \dot q)}_{k+\frac12}\right).
\end{align*}
Here the indices correspond to the discrete time evolution. The discrete time
evolution is then obtained by subsequently solving the discrete
Euler-Lagrange equations
\begin{align*}
D_2 L^d(q_{k-1}, q_k) + D_1 L^d(q_k, q_{k+1}) = 0,
\end{align*}
for $q_{k+1}$, with $q_{k-1}$ and $q_k$ known.
In order to evaluate the Lagrangian $L$~\eqref{eqn:lagrangian} we
discretize the system as follows: We replace the smooth body boundaries
by polygons, and represent the potential fields $u_I$ and $u_\Bdy$ using point
sources, which are attached to the rigid bodies. This discretization has
previously been used to compute the Kirchhoff tensor of 3D bodies
\citep{Weissmann2012}. This allows to explicitly compute all quantities and
variations needed for evaluating the discrete Euler-Lagrange equations, and we
have implemented a numerical scheme in this way. For validation we have
simulated the following configurations:

\begin{figure}
\centering
\includegraphics[width=0.33\columnwidth]{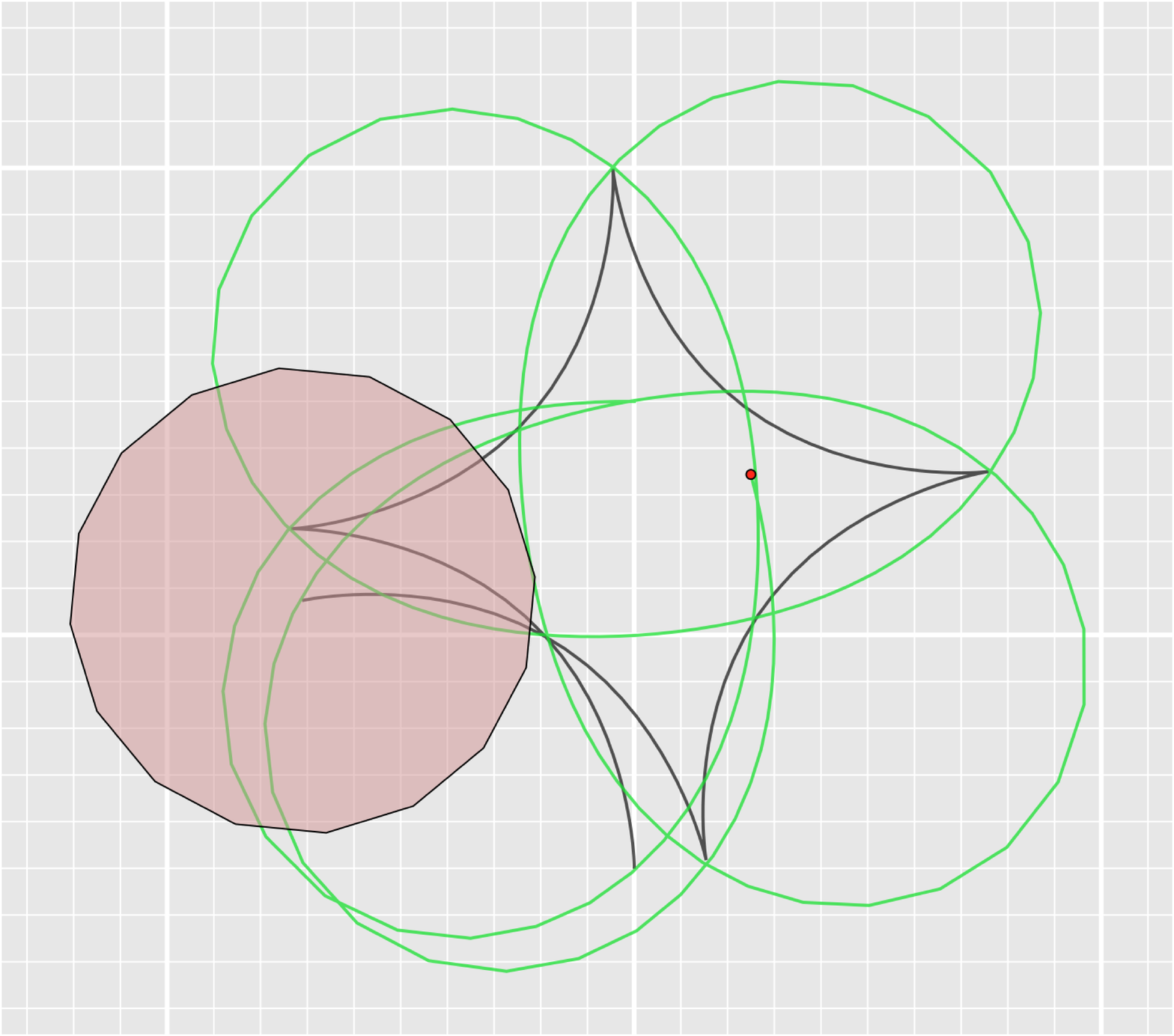}
\hspace{0.0025cm}
\includegraphics[width=0.33\columnwidth]{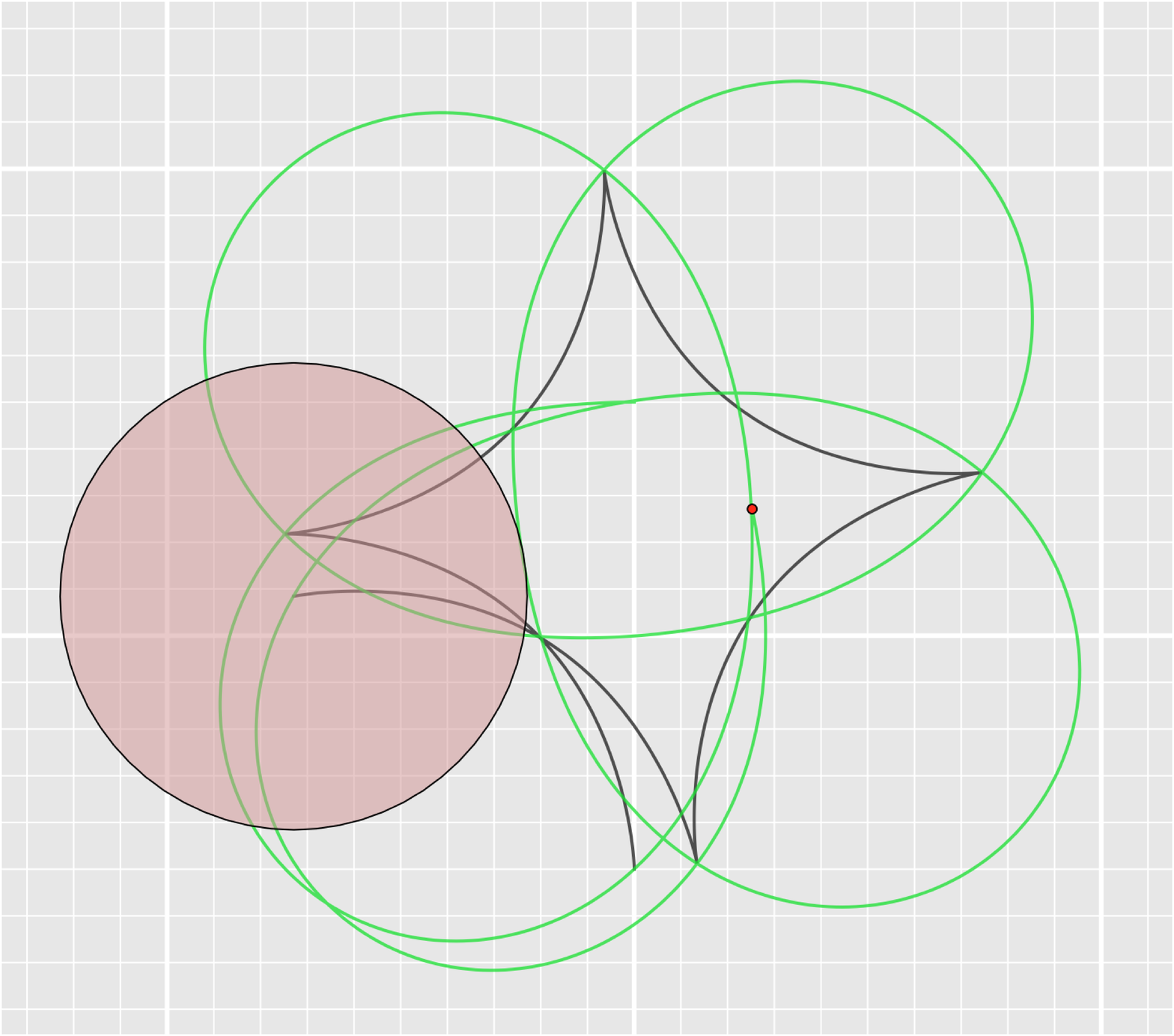}\\[0.125cm]
\includegraphics[width=0.33\columnwidth]{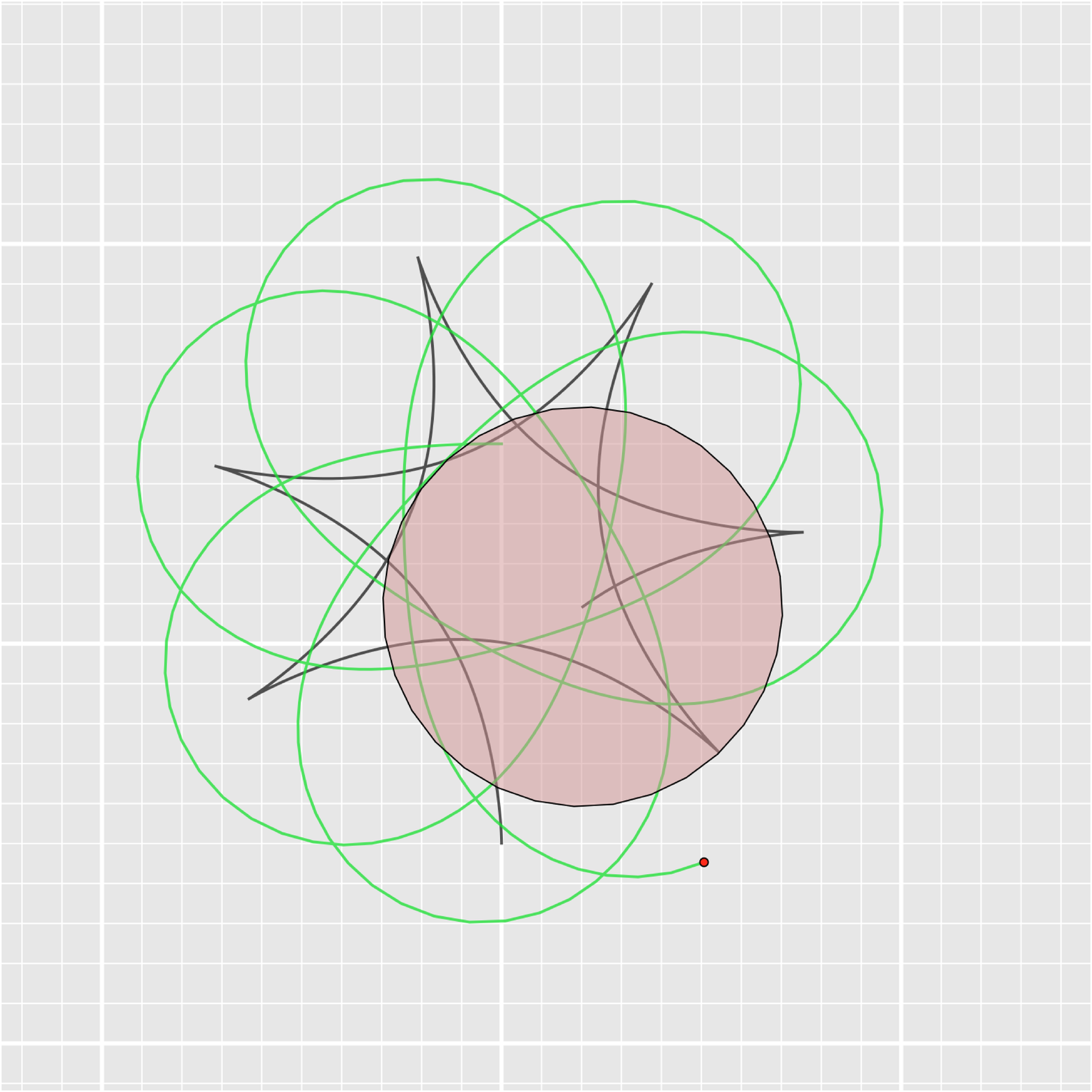}
\hspace{0.0025cm}
\includegraphics[width=0.33\columnwidth]{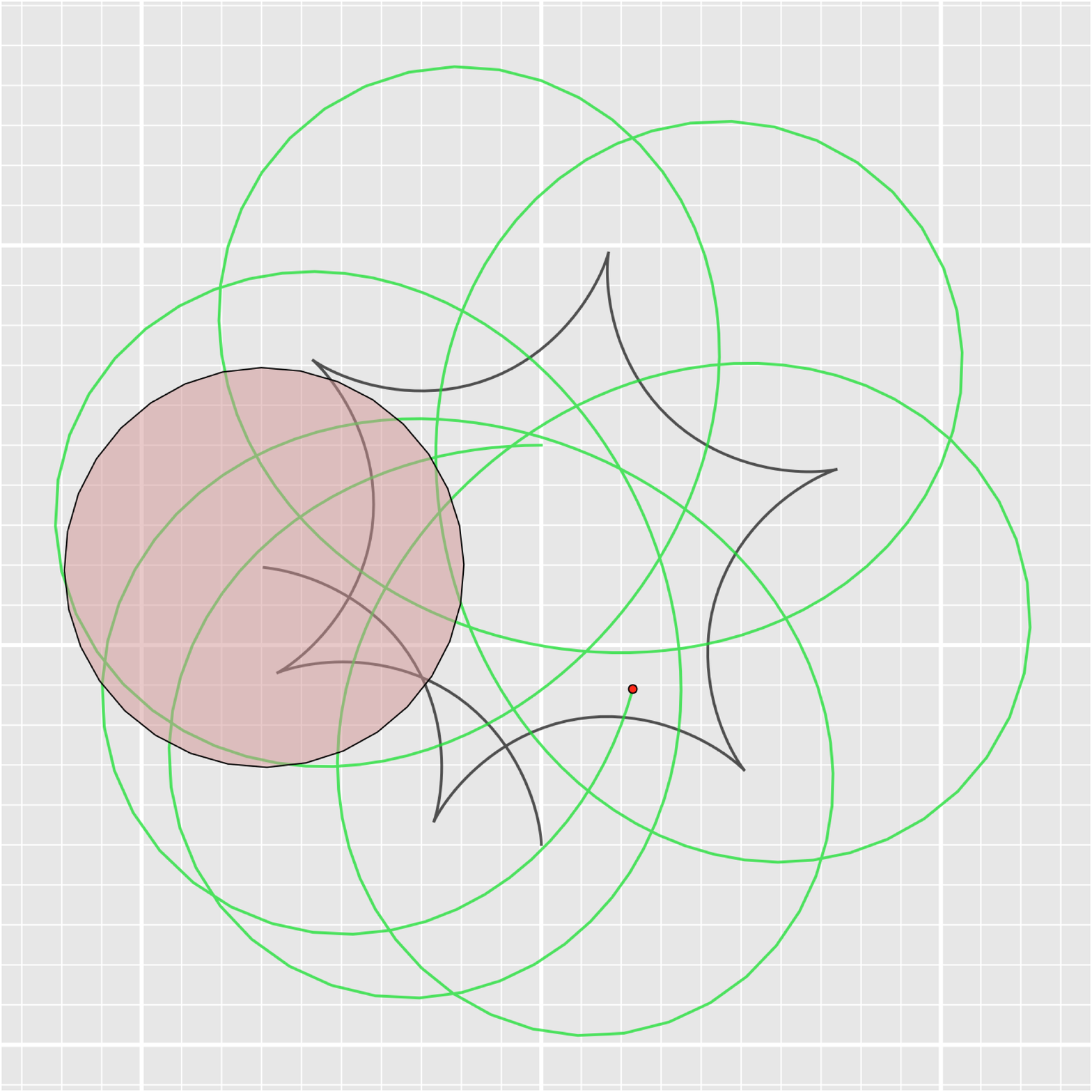}
\caption{Interaction of a unit disk with one point vortex.
\emph{Top:} The two
figures show the simulation of the same initial setup, but with different time
step (0.025 vs. 0.25) and body discretization (16 vs. 64 edges). Note that
the overall structure is identical, even though both discretization and time step
are very different.
\emph{Bottom:} Simulations with a light (density 0.1) and a heavy (density 4.0)
disk.
Mass affects the frequency of the periodic motion.}
\label{fig:disk-vortex}
\end{figure}
\begin{figure}
\centering
\includegraphics[width=0.65\columnwidth]{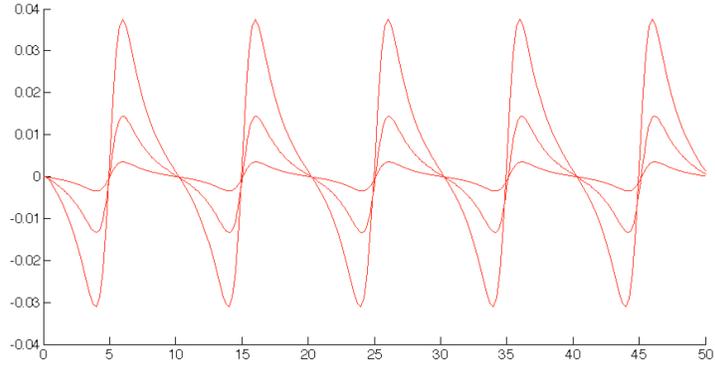}
\caption{Energy oscillation of the disk-vortex interaction. The three plots correspond
to time steps 0.025, 0.1 and 0.25. Note that the oscillations agree precisely,
but the magnitude is proportional to the time step.}
\label{fig:energy-disk-vortex}
\end{figure}

\begin{figure}
\centering
  \includegraphics[width=.325\textwidth]{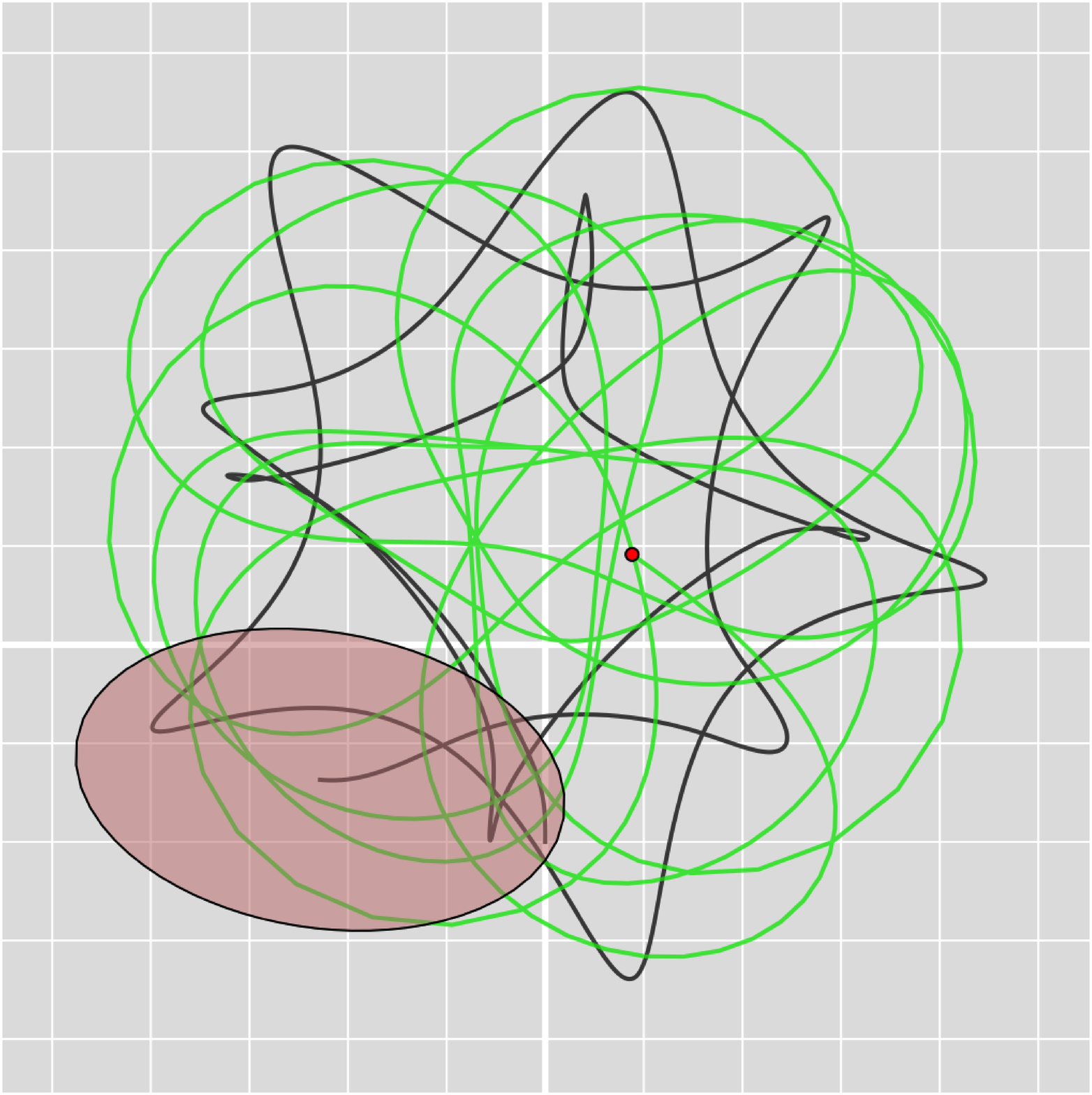}\,\,\,
 \includegraphics[width=0.6\columnwidth]{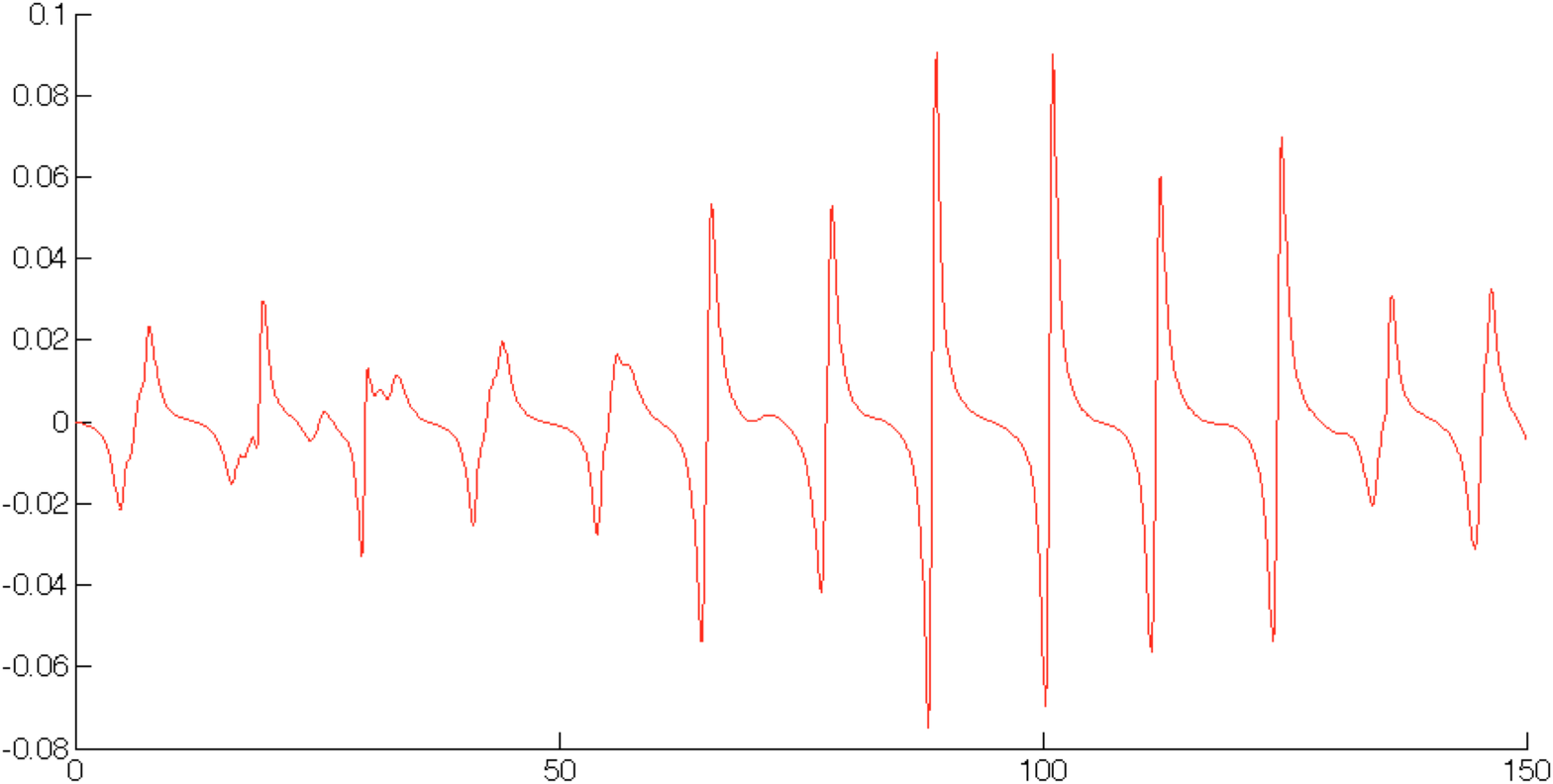}
  
\caption{Interaction of a point vortex with an ellipse, with time step
0.1. The motion is chaotic, in contrast to disk/vortex interaction.
The plot shows energy oscillations. Large peeks correspond to high
dynamical interaction, when body and point vortex are close together.
}
\label{fig:ellipse-vortex}
\end{figure}


\begin{figure}
\centering
\includegraphics[width=0.425\columnwidth]{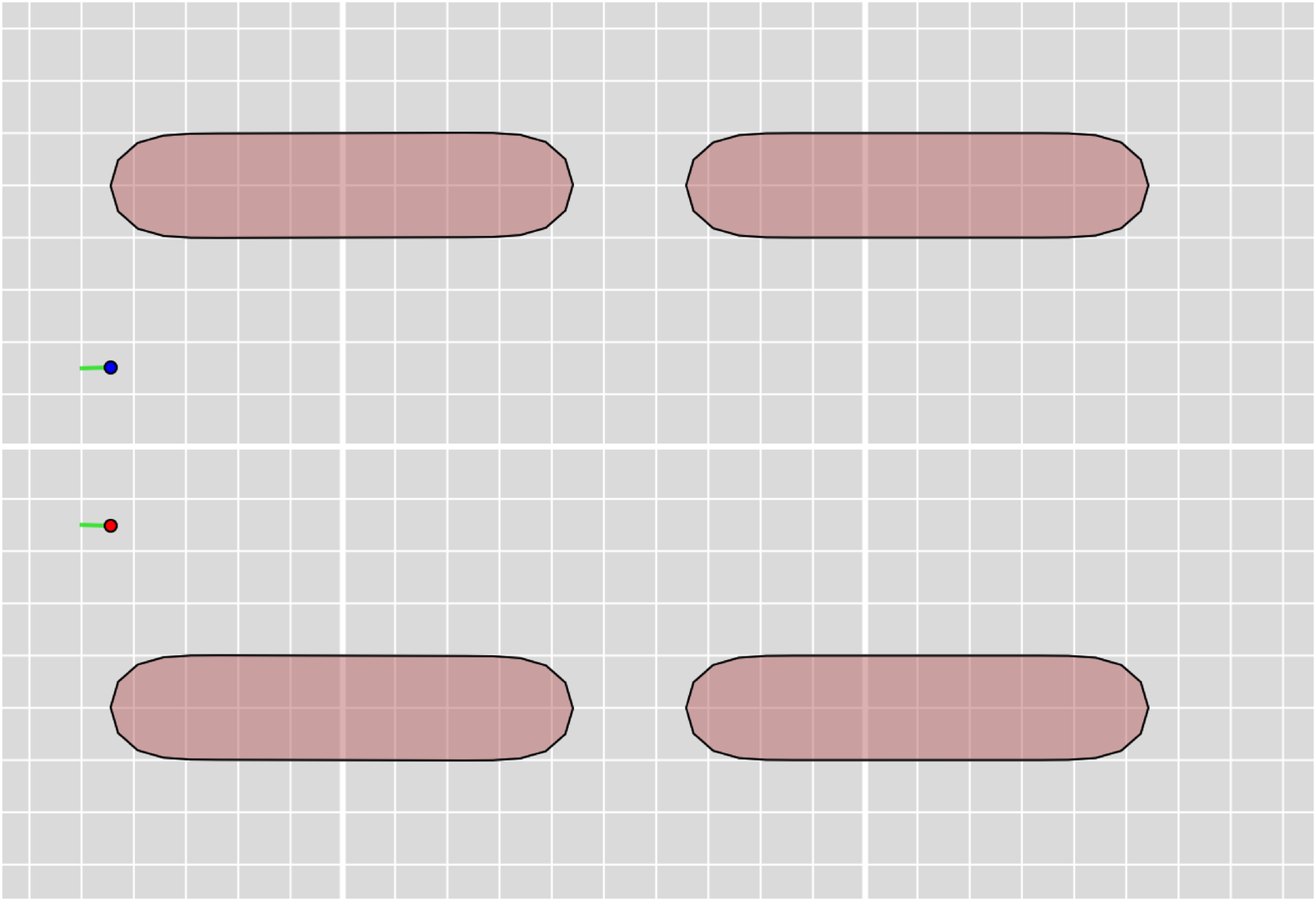}\,\,\,
\includegraphics[width=0.425\columnwidth]{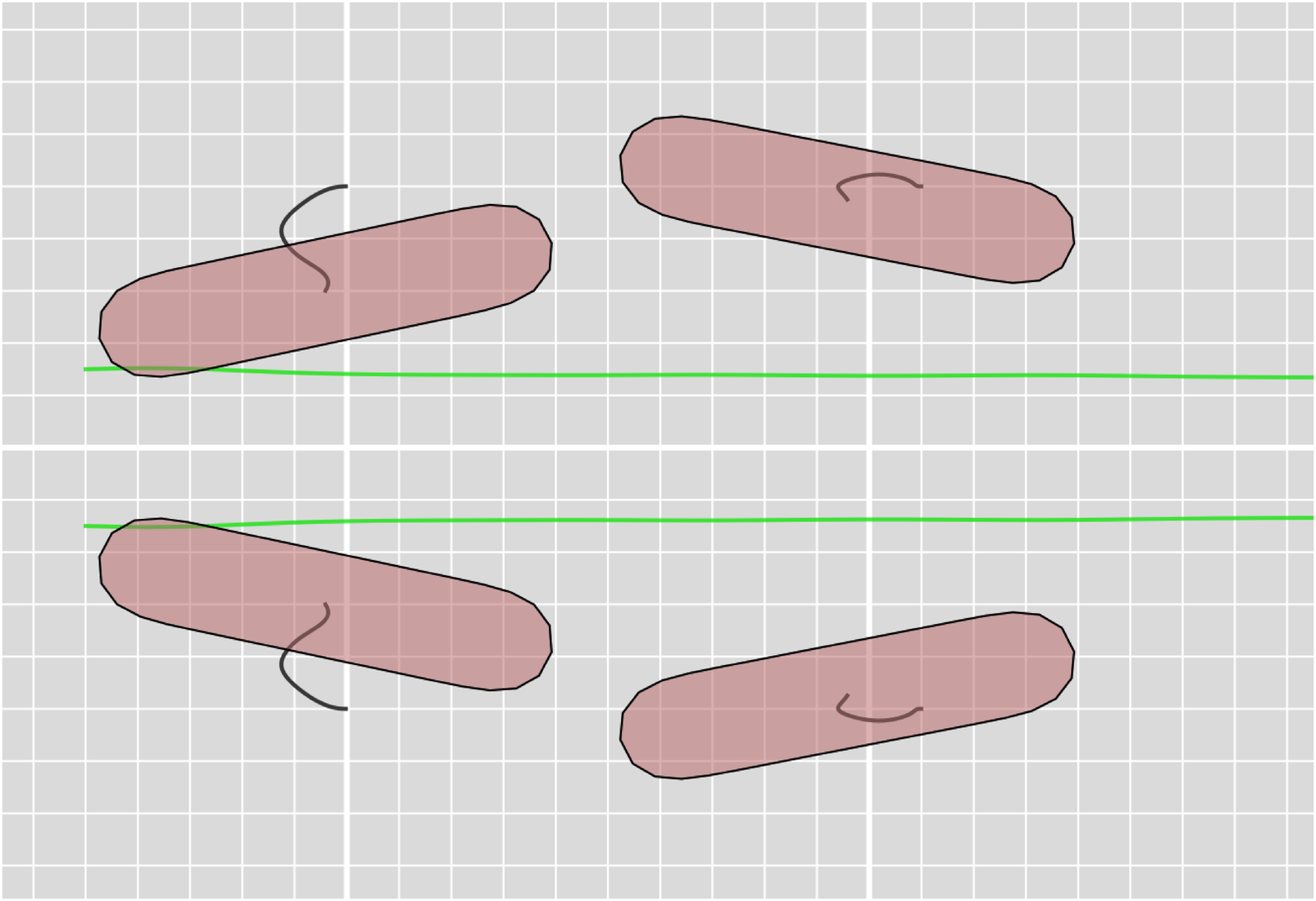}
\caption{Fluid flow inside of a channel, initial configuration (left) and
configuration after the vortex pair has traveled through the channel (right).
Because of the higher velocity the pressure is lower in between the walls
(Bernoulli's principle), dragging them together.}
\label{fig:channel}
\end{figure}

\textbf{Disk with single point vortex:} This case is particularly interesting,
since it is one of the rare cases where fluid-body interaction is integrable
\citep{Borisov2003-int}. The system has periodic (and even closed) orbits, i.e.,
disk and vortex ``dance'' around each other, producing a regular, periodic
pattern. The mass of the disk determines the frequency of the oscillating
motion. Different motions are shown in Figure~\ref{fig:disk-vortex}.

\textbf{Ellipse with single point vortex:}
This case can be viewed as a distortion of the integrable case of a disk.
However, this small change drastically changes the behavior of the system:
There are no closed orbits and the motion is chaotic, see Figure~\ref{fig:ellipse-vortex}.


\textbf{Flow through a channel:}
This configuration (Figure~\ref{fig:channel}) illustrates Bernoulli's principle,
i.e., pressure is low in regions of high velocity. A vortex pair travels
through a channel made out of four flat objects. Due to higher velocity
in between the walls are pulled together.

All simulations preserve linear and angular momentum (in the absence of external
forces) up to the precision used when solving the discrete Euler-Lagrange
equations.
The total energy of the system oscillates around its true value during the
simulations. The magnitude of these oscillations appears to be proportional to
the chosen time step (Figure~\ref{fig:energy-disk-vortex}), while the body
discretization has no significant influence. These oscillations can be large at
times of high dynamical interaction, i.e., when the vortices are very close to
the bodies. Nevertheless there is no drift, only oscillations around the true
energy level (Figure~\ref{fig:ellipse-vortex}, right). All simulations were
computed on a Macbook Pro with a 2.7\,GHz Intel Core i7 and 16\,GB RAM. The
implementation is done in Matlab, and uses no performance optimization such as
GPU computations.
Configurations with one body take about 0.5\,s per time step, the channel
example (Figure~\ref{fig:channel}) with 4 bodies around 30\,s per time step.

\section{Conclusions and Outlook}

We have introduced the Hamiltonian description for several rigid
bodies interacting with point vortices, assuming zero circulation around the
individual bodies, but arbitrary point vortex strengths. We have used the general
framework of cotangent bundle reduction only to determine the reduced phase
space of the system, as well as the general structure of the symplectic form on
the reduced phase space. From there we have determined the symplectic
form directly, without resorting to the abstract framework of mechanical
connections.
From the Hamiltonian formulation we have given a Lagrangian description of
the dynamics, and derived a variational time integrator following
\citet{Marsden2001} and \citet{Rowley2002}. Using polygonal bodies and point
sources, we have implemented a numerical algorithm to simulate the coupled
dynamics and validated the implementation with different configurations.\\
We expect that our formulation generalizes to the 3D case, describing the
dynamics of several rigid bodies interacting with vortex filaments.
So far the dynamics is only known for the case of a single rigid body
\citep{Shashikanth2008}.\\

\textbf{Acknowledgment:}
Ulrich Pinkall proposed the basic idea for deriving the symplectic form. It is
my great pleasure to thank him for invaluable discussions and suggestions. Felix
Kn{\"o}ppel and David Chubelaschwili helped working out many of the details. Eva
Kanso and the anonymous reviewers provided important feedback for improving the
exposition. This work is supported by the DFG Research Center Matheon and the
SFB/TR 109 ``Discretization in Geometry and Dynamics''.

\appendix

\section{Cotangent Bundle Reduction of Fluid-Body Dynamics}

\label{app:cotangent-bundle-reduction}

In analogy to Arnold's geometric description of fluid
dynamics~\citep{Arnold1966}, the dynamics of rigid bodies interacting with a
surrounding incompressible fluid can be viewed as a geodesic problem on a
Riemannian manifold. The kinetic energy defines a Riemannian metric on the
configuration space, and geodesics satisfy Hamilton's equations on the cotangent
bundle with kinetic energy as the Hamiltonian. This insight is due to
\citet{Vankerschaver2009} (VKM). The authors use the framework of cotangent
bundle reduction \citep{Marsden2007} to obtain a reduced Hamiltonian system with
magnetic symplectic form for the case of a single body in a fluid whose
vorticity field is concentrated at point vortices.

The Hamiltonian formulation by VKM is an extension of Arnold's original work
\citep{Arnold1966}, which describes the motion of an incompressible inviscid
fluid in a fixed fluid domain $\Fld$ as a geodesic on the group
$\Diff_{\vol}(\Fld)$ of volume-preserving diffeomorphisms on $\Fld$.
However, when the fluid interacts with rigid bodies, the fluid domain is no
longer fixed. The idea of VKM is to consider the space $\Emb_{\vol}(\Fld^0,
\R^2)$ of volume-preserving embeddings of an initial reference configuration
$\Fld^0$ into $\R^2$ instead of $\Diff_{\vol}(\Fld)$. Any incompressible fluid
motion is then described by a curve in the subset
$\mathcal{Q}^\Fld \subset \Emb_{\vol}(\Fld^0, \R^2)$ which is compatible with
the body motion. The configuration space of the coupled system is $\mathcal{Q} =
\SE(2)^n \times \mathcal{Q}^\Fld$, and the dynamics is a canonical Hamiltonian
system on $T^* \mathcal{Q}$ with kinetic energy as the Hamiltonian.

The kinetic energy is invariant under volume-preserving diffeomorphisms of the
initial fluid configuration $\Fld^0$ (\emph{particle relabeling symmetry}),
i.e., the symmetry group $\Diff_{\vol}(\Fld^0)$ acts from the right on
$\mathcal{Q}^\Fld$, and thus on $\mathcal{Q}$. This action turns $\mathcal{Q}$
into a principal fiber bundle over $\SE(2)^n$. This structure allows to follow
the famous Kaluza-Klein approach to determine the Hamiltonian dynamics. In order
to factor out the $\Diff_{\vol}(\Fld^0)$-symmetry one needs to fix a value of
the associated momentum map, which corresponds to choosing an initial vorticity
field of the fluid. This is where the assumption is used that vorticity is
concentrated at $m$ point vortices. The reduced phase space is $\mathcal{M} =
T^* SE(2)^n \times \R^{2 m}$, see VKM, \S~4.2, and the dynamics is given by a
reduced symplectic form $\sigma$ on $\mathcal{M}$ with kinetic energy as the
Hamiltonian. The following theorem formulates the starting point for the
derivations made in this paper.\\

\begin{theorem}
The dynamics of $n$ rigid bodies interacting with $m$ isolated point vortices
is a Hamiltonian system. The Hamiltonian is the
kinetic energy~\eqref{eqn:kinetic-energy}, and the phase space is
\begin{align*}\mathcal{M} = T^* \SE(2)^n \times \R^{2 m}.
\end{align*}
The cotangent bundle $T^* \SE(2)^n$ corresponds to the rigid body configuration,
and $\R^{2 m}$ is the phase space for $m$ point vortices. The symplectic form is
\begin{align*}
 \sigma = \sigma_{can} + \d
\alpha + \sigma_\gamma,
\end{align*}
where $\sigma_{can}$ is the canonical symplectic form
on the cotangent bundle $T^*\SE(2)^n$,  $\sigma_\gamma$ is the Kirillov-Kostant-Sariou form on the
coadjoint orbit $\R^{2 n}$, and and $\d \alpha$ is a magnetic
term, i.e., a two-form on $SE(2)^n \times \R^{2 m}$.
\end{theorem}
\begin{proof} This has been proven in VKS, \S 4. We emphasize here that the
proofs do not rely on the fact that only a single rigid body was considered.
\end{proof}

\section{The Cotangent Bundle of Euclidean Motions}
\label{app:cotangent-bundle}

In this section we consider the Lie group of Euclidean motions $\SE(2)$ and
denote the pairing between covectors and vectors by $\pair{., .}$. For any
covector $\mu \in T_g^* \SE(2)$ we can find a body momentum $M \in \R^3 \cong
\se^*(2)$ such that $\pair{\mu, \delta g} := \tdd{M, \Lambda}$, for any $\delta
g = g \Lambda$. Note that $\Theta(\delta \mu, \delta g) := \pair{\mu, \delta g}$
is a one-form on the contangent bundle $T^*\SE(2)$, and $\tdd{M, \Lambda}$ is
its push-forward to the \emph{left trivialization} $\R^3 \times \SE(2) \cong
T^*\SE(2)$. It is the \emph{canonical one-form}, and its exterior derivative
gives the canonical symplectic form on $T^*\SE(2)$. We will now compute the
symplectic form when pushed forward to the left trivialization $\R^3 \times
\SE(2)$, using the general formula for the exterior derivative of a one-form:
\begin{align}
\label{eqn:d-omega}
\d \Theta(X, Y) = \nabla_X \Theta(Y) - \nabla_Y \Theta(X) - \Theta([X, Y]).
\end{align}
Here $X$ and $Y$ are vector fields and $[X, Y]$ is the Jacobi-Lie bracket of $X$
and $Y$. Consider a two-parameter family $(M(s,t), g(s,t))$ in $\R^3 \times
\SE(2)$, whose partial derivatives (denoted by $\delta$ and $'$, respectively)
commute. The vector fields will be $X = (\delta M, \delta g)$ and $Y = (M',
g')$, where $\delta g = g \Gamma$ and $g' = g \Xi$ with $\Xi = (\Omega, V)$. One
can check that the partial derivatives of $g$ commute if and only if
\begin{align}
\label{eqn:velocity-integrability}
\Gamma' & = \delta \Xi - \ad{\Xi} \Gamma, & \ad{\Xi} & = \mat{0 & 0 \\ V \times &
\Omega \times}.
\end{align}
The commuting partial derivatives ensure that the Jacobi-Lie bracket
in~\eqref{eqn:d-omega} vanishes. The covariant derivatives are usual directional
derivatives here, so we obtain the canonical symplectic two--form $\sigma = \d
\Theta$ in the left-trivialization as
\begin{equation}
\label{eqn:canonical-symplectic-form}
\sigma\left((\delta M, \Gamma), (M', \Xi)\right) =
\delta{\dd{M, \Xi}} - \dd{M, \Gamma}' 
 = \dd{\delta M, \Xi} - \dd{M' - \coad{\Xi} M.
\Gamma},
\end{equation}
Here $\coad{\Xi}$ is the matrix transpose of $\ad{\Xi}$:
\begin{align}
\label{eqn:algebra-coadjoint-operator}
\coad{\Xi} = - \mat{0 & V \times \\ 0 & \Omega \times}.
\end{align}

\bibliographystyle{apalike}
\bibliography{multibody2d}

\end{document}